%% file: main.tex
\documentclass[conference]{IEEEtran}
\IEEEoverridecommandlockouts

\usepackage{cite}
\usepackage{amsmath,amssymb,amsfonts}
\usepackage{algorithmic}
\usepackage{graphicx}
\usepackage{textcomp}
\usepackage{xcolor}
\usepackage{amsmath}
\usepackage{color}
\usepackage{xcolor}
\usepackage{url}
\usepackage{xspace}
\usepackage{tikz}
\usepackage{times}
\usepackage{graphicx}
\usepackage{subcaption}
\usepackage{algorithm}
\usepackage{algorithmic}
\usepackage[algo2e,boxruled]{algorithm2e}
\usepackage{multicol}
\usepackage[flushleft,alwaysadjust]{paralist}
\usepackage{listings}
\usepackage[compact]{titlesec}		
\usepackage{hyperref}
\usepackage[capitalize]{cleveref}

\def\BibTeX{{\rm B\kern-.05em{\sc i\kern-.025em b}\kern-.08em
    T\kern-.1667em\lower.7ex\hbox{E}\kern-.125emX}}

\long\def\com#1{}

\newcommand{\sys}{\textsf{RACS}\xspace}
\newcommand{\mempool}{\textsf{SADL}\xspace}

\usepackage{amsthm}

\theoremstyle{plain}

\newtheorem{theorem}{Theorem}

\begin{document}

\title{RACS-SADL: Robust and Understandable Randomized Consensus in the Cloud}

\author{\IEEEauthorblockN{Pasindu Tennage}
\IEEEauthorblockA{
EPFL and ISTA\\
Lausanne, Switzerland \\
pasindu.tennage@epfl.ch}
\and
\IEEEauthorblockN{Antoine Desjardins}
\IEEEauthorblockA{
ISTA\\
Klosterneuburg, Austria \\
antoinedesjard@gmail.com}
\and
\IEEEauthorblockN{Lefteris Kokoris-Kogias}
\IEEEauthorblockA{
Mysten Labs and ISTA\\
Athens, Greece \\
lefteris@mystenlabs.com}
}

\maketitle

\begin{abstract}
\input{abstract}
\end{abstract}


\input{introduction}
\input{background}
\input{racs-design}
\input{optimizations}
\input{implementation}
\input{evaluation}

\input{conclusion}

\bibliographystyle{IEEEtran}
\bibliography{references}
\newpage

\appendix
\input{racs_proof}

\end{document}

%% file: abstract.tex
Widely deployed consensus protocols in the cloud 
are often leader-based and optimized for low 
latency under synchronous network conditions. 
However, cloud networks can experience 
disruptions such as network partitions, high-loss 
links, and configuration errors. These 
disruptions interfere with the 
operation of leader-based protocols, as their 
view change mechanisms interrupt the normal case replication and cause the system to stall.

We propose \sys, a novel randomized consensus protocol that ensures robustness against adversarial network conditions.
\sys achieves optimal one-round trip latency under synchronous network conditions while remaining resilient to adversarial network conditions.
\sys follows a simple design inspired by Raft, the most widely used consensus protocol in the cloud, and therefore enables seamless integration with the existing cloud software stack.

Experiments with a prototype running on Amazon EC2 show that \sys achieves 28k cmd/sec throughput, ninefold higher than Raft under adversarial cloud network conditions.
Under synchronous network conditions, \sys matches the performance of Multi-Paxos and Raft, achieving a throughput of 200k cmd/sec with a median latency of 300ms, confirming that \sys introduces no unnecessary overhead.
Finally, \mempool-\sys, a throughput-optimized version of \sys, achieves a throughput of 500k cmd/sec, delivering 150\% higher throughput than Raft.

%% file: introduction.tex
\section{Introduction}\label{sec:intro}

Consensus~\cite{cachin2011introduction}
enables a set of replicas to maintain strongly consistent
state while remaining resilient to failures, 
and is widely used in
distributed cloud applications~\cite{burrows2006chubby},\cite{hunt2010zookeeper},\cite{maccormick2004boxwood},\cite{corbett2013spanner},\cite{baker2011megastore},\cite{xie2014salt},\cite{quintero2011implementing},\cite{mashtizadeh2013replication},\cite{grimshaw2013gffs},\cite{bronson2013tao},\cite{lloyd2011don}.

Most deployed consensus protocols in the cloud follow a leader-based design~\cite{hunt2010zookeeper,maccormick2004boxwood,lamport2001paxos,ongaro2014search}.
The leader replica replicates client commands among the non-leader replicas and responds back to clients.
Under normal case, where the network is synchronous, leader-based protocols provide moderate throughput and low latency.

Modern cloud networks~\cite{nastic2020sloc,ding2019characterizing} use high-speed connections and are often synchronous.
However, cloud networks can occasionally encounter adverse conditions such as high latency, packet loss, and partial connectivity~\cite{cloudflare-nov20}. 
These issues often stem from failing network hardware, congested links, inaccuracies in failure detection software, and flaws in monitoring systems~\cite{moura2016review,de2022noise}.
Hereafter, we refer to these occasional adverse cloud network conditions as "adversarial network conditions".

In practice, leader-based protocols and their multi-leader extensions~\cite{moraru2013there}\cite{barcelona2008mencius}\cite{charapko2021pigpaxos}\cite{lamport2005generalized} are designed to detect adversarial network conditions using timeouts~\cite{ongaro2014search}.
If the leader replica does not respond within a pre-configured timeout, the remaining replicas elect a new leader replica.
Electing a new leader replica often involves a separate "view change" sub routine~\cite{ongaro2014search}, where each replica votes for a potential next leader replica.

View change-based leader election introduces performance and robustness challenges~\cite{tennage2022baxos,pasinduQuePaxa2023}.
During view change, no new client commands are committed, resulting in a prolonged commit delay for commands submitted during this period.
Client commands proposed during the view change accumulate into a substantial backlog, causing significant overhead that persists even after the completion of the view change-based leader election.
Since each replica buffers client commands during the view change, a significant backlog of commands accumulates by the time the leader election is completed.
The large backlog of buffered commands places significant strain on the next elected leader, leading to further performance degradation.

Asynchronous consensus protocols~\cite{ ben1983another}\cite{BoBandle} ensure high robustness and performance under adversarial network conditions by leveraging randomization instead of relying on view change.
However, to the best of our knowledge, asynchronous consensus protocols have never been deployed in cloud applications due to three main drawbacks:  
(1) they require significantly higher bandwidth compared to leader-based protocols~\cite{ben1983another};  
(2) their design is often incompatible with the existing cloud software stack, hindering seamless integration and rapid adoption; and  
(3) their complexity makes them difficult to understand and implement, which limits industry adoption.
Hence, cloud providers continue to deploy leader-based consensus protocols despite their limited robustness against adversarial network conditions~\cite{burrows2006chubby}.

\textbf{Problem Statement:}  
Can we design a consensus protocol that (1) remains robust under adversarial network conditions, (2) achieves performance comparable to leader-based protocols under synchronous network conditions, and (3) maintains a simple design that allows easy integration with the existing cloud software stack?

This paper proposes \sys, a robust and easily understandable randomized consensus protocol that addresses the above problem. 
\sys achieves three key goals: 
\begin{compactitem}
\item \textbf{G1}: Provides robustness against adversarial networks.
\item \textbf{G2}: Matches the performance of high speed leader-based protocols~\cite{lamport2001paxos,ongaro2014search} under synchronous network conditions. 
\item \textbf{G3}: Employs an easy-to-understand design that seamlessly integrates with existing cloud software stack, enabling rapid adoption in the cloud.
\end{compactitem}

Achieving robustness against adversarial network conditions has already been explored~\cite{ben1983another, danezis22narwhal, pasinduQuePaxa2023, BoBandle}, and we explicitly state that this is "not" our primary contribution.
Instead, \sys, our novel protocol, "simultaneously" achieves the three goals G1, G2, and G3—an accomplishment that no existing asynchronous protocol has achieved before.
Designing \sys while achieving G1, G2 and G3 simultaneously is not trivial.
First, designing a protocol that ensures robustness against adversarial network conditions without incurring additional overhead under synchronous conditions is a non-trivial algorithmic challenge.
Second, designing a robust protocol while ensuring ease of understanding and compatibility with existing cloud software stacks is a challenge that no previous asynchronous consensus protocol has successfully addressed.

We implemented and evaluated \sys 
in Go \cite{go-lang} and compared it against the existing implementations of
Multi-Paxos\cite{lamport2001paxos}, Raft\cite{ongaro2014search},
and EPaxos \cite{moraru2013there}.
We evaluated \sys
on Amazon EC2 in a multi-region WAN setting.
First, we show that \sys provides 28k cmd/sec  
of throughput, under adversarial 
network conditions,
and outperforms Multi-Paxos and Raft which only provide 2.8k cmd/sec
in the same setting.
Second, we show that \sys delivers 200k cmd/sec in throughput
under 300ms median latency, comparable to Multi-Paxos's 
200k cmd/sec throughput, under synchronous normal case conditions.
Finally, we introduce novel optimizations for \sys and demonstrate that the optimized version; \mempool-\sys achieves an impressive throughput of 500k cmd/sec.

\subsection{Contributions}

\begin{compactitem}
    \item We propose \sys, the first practical randomized consensus protocol that simultaneously achieves (1) robustness against adversarial network conditions, (2) high performance under typical synchronous network conditions, and (3) seamless integration with existing cloud software stack.
    \item We provide the formal proofs of \sys in~\Cref{sec:formal-proof}.
    \item We propose \mempool, an optimization for \sys that enables high throughput and high robustness.
    \item We implement a prototype of \sys and conduct an experimental analysis on Amazon EC2.
    \item We show that \sys achieves nine times greater robustness than Raft under adversarial network conditions, while \mempool-\sys delivers 150\% higher throughput than Raft. 
\end{compactitem}

%% file: background.tex
\section{Background and Related Work}\label{sec:background}

\subsection{Threat Model and Assumptions}\label{subsec:threat-model}


We consider a system with $n$ replicas.
Up to $f$ (where $n$$=$2$f$+1) number of replicas can crash,
but replicas do not equivocate nor commit
omission faults \cite{cachin2011introduction}.


We assume first-in-first-out (FIFO)
perfect point-to-point links \cite{cachin2011introduction} between
each pair of replicas. 
In practice, TCP\cite{rfc793} provides FIFO perfect point-to-point links.
We assume a content-oblivious~\cite{aspnes2003randomized}
network adversary; the adversary 
may manipulate network delays,
but cannot observe
the message content nor 
the internal replica state.
In practice, 
TLS~\cite{rescorla2018transport}
encrypted channels 
between each pair of replica
satisfy this assumption.

FLP~\cite{fischer1985impossibility} impossibility result claims that no deterministic protocol can guarantee both safety
and liveness in a fully asynchronous system with even one faulty process.
Practical protocols address the FLP~\cite{fischer1985impossibility} impossibility using two approaches: (1) assuming partial synchrony, 
and (2) using randomization. 
Protocols like Raft~\cite{ongaro2014search} and Multi-Paxos~\cite{lamport2001paxos} assume partial synchrony and ensure safety at all times, however, lose liveness under asynchrony.
Randomized protocols, in contrast, employ randomness instead of determinism, preserving safety always and ensuring liveness under asynchrony with probability 1.

In a typical cloud network, there are periods
in which the network behaves synchronously,
followed by phases where the network shows adversarial behavior.
\sys makes use of this network behavior
and operates in two modes:
(1) synchronous mode and (2) fallback mode.
During the synchronous periods, \sys employs a leader-based design
to reach consensus using one round trip
network delay, and during the fallback mode, \sys employs randomization.

\subsection{Consensus}\label{subsec:consensus-intro}

Consensus enables a set of replicas to reach
an agreement on a single history of values.
A correct consensus algorithm satisfies four
properties \cite{cachin2011introduction}:
(1) \textit{validity}: the agreed upon value should be
previously proposed by a replica,
(2) \textit{termination}: every correct process
eventually decides some value,
(3) \textit{integrity}: no process decides twice,
and (4) \textit{agreement}: no two correct processes decide differently.

\begin{figure}[t]
    \centering
    \includegraphics[scale=0.25]{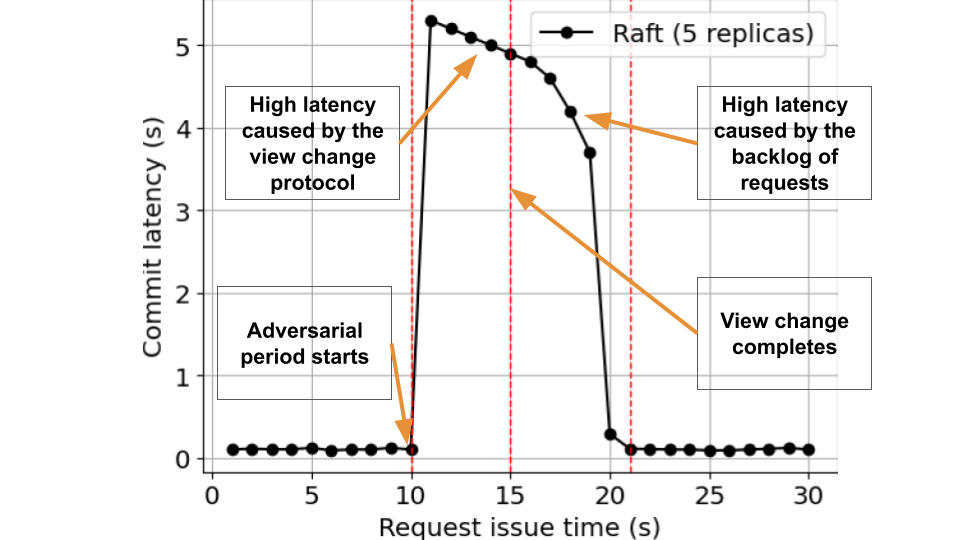}
    \caption{High latency overhead due to view change in Raft.}
    \label{fig:hangover}
    \vspace{-5mm}    
\end{figure}

\textbf{Robustness problem in leader-based protocols}: 
Multi-Paxos \cite{lamport2001paxos}, Raft \cite{ongaro2014search}, and View Stamp Replication \cite{oki1988viewstamped} employ a single leader replica to totally order commands.
Under normal network conditions, the leader replica replicates client commands 
in one network round trip.
When the network conditions become adversarial, leader-based protocols
run a separate "view-change" sub-routine that elects a new leader replica.
View change introduces robustness challenges in leader-based protocols, which we illustrate using~\Cref{fig:hangover}.

In~\Cref{fig:hangover}, we run Raft—the most widely deployed consensus protocol in the cloud—with 5 replicas, with the view timeout configured at 5 seconds. 
At time t = 10s seconds, we simulate an adversarial network condition by introducing a 6-second delay on the links connected to the current leader.
We first observe that the commands issued between \(10\)s and \(15\)s experience significantly high commit latency.
This high latency arises from the view change, as no new commands are committed until the view change is fully completed.
Second, we observe that commands issued after the completion of the view change also experience high commit latency.
This is due to the backlog of commands that accumulated during the view change.
Therefore, existing leader-based consensus protocols are unable to ensure robustness under adversarial network conditions.
In contrast, \sys, our novel protocol, ensures robustness under adversarial network conditions.

\textbf{Multi-leader protocols}~\cite{barcelona2008mencius, moraru2013there, lamport2005generalized, tennage2022baxos} and other extensions of leader-based protocols~\cite{marandi2010ring}\cite{ng23omni}\cite{charapko2021pigpaxos}\cite{ailijiang2017wpaxos} also lose liveness under adversarial network conditions due to their reliance on view change mechanisms. 

\textbf{Randomized consensus protocols}~\cite{BoBandle}\cite{ ben1983another} overcome the adversarial network challenge by using randomization instead of view changes. 
However, they are rarely used in cloud applications due to  
(1) high performance overhead under synchronous network conditions,  
(2) incompatibility with the existing cloud software stack, and  
(3) complexity in understanding and implementing.
In contrast, our protocol \sys provides a practical solution to adversarial network challenges by achieving:  
(1) robustness under adversarial network conditions,  
(2) optimal one round-trip commit latency under synchronous networks, and 
(3) a simple design for fast implementation and seamless industry adoption.

Rabia~\cite{pan2021rabia} and \sys take opposite paths on the simplicity–robustness tradeoff. Rabia~\cite{pan2021rabia} aims for simplicity in synchronous LANs but loses liveness under WAN asynchrony. In contrast, \sys ensures liveness under asynchrony. Rabia’s~\cite{pan2021rabia} leaderless design differs from Paxos/Raft standard leader-based model, requiring ground-up changes in the implementation, making adoption costly. Rabia~\cite{pan2021rabia} has quadratic message complexity, whereas \sys has linear message complexity, under synchrony.

QuePaxa~\cite{pasinduQuePaxa2023} solves the asynchronous consensus challenge, however, significantly differs from \sys in design.
QuePaxa~\cite{pasinduQuePaxa2023} uses randomized local-coins and splits the replica role into separate proposer role and recorder role. 
In contrast, \sys employs global-coins and closely follows the standard Raft model.

Mahi-Mahi~\cite{mahi-mahi}, Tusk~\cite{danezis22narwhal}, and Dag-Rider~\cite{keidar2021all} are “blockchain” asynchronous protocols targeting Byzantine faults, which is out of our scope. 
We omit detailed comparisons due to space and relevance.

%% file: racs-design.tex
\section{\sys}\label{sec:design-racs}

\subsection{\sys Protocol Overview}

\begin{figure}[t]
    \centering
    \includegraphics[width=0.89\linewidth, height=0.09\paperheight]{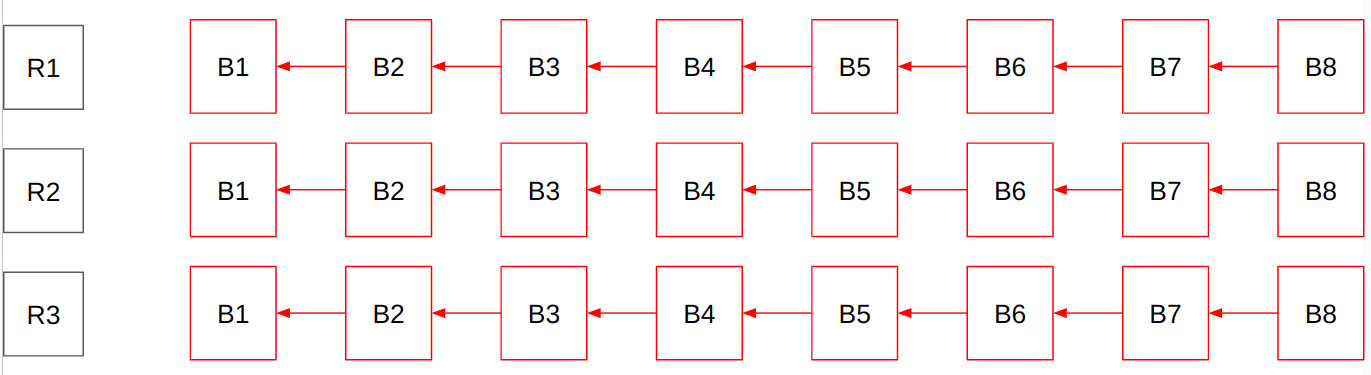}
    \caption{Replicated log of \sys replicas.}
    \label{fig:blocks}
    \vspace{-5mm}    
\end{figure}

\begin{figure}[t]
    \centering
    \includegraphics[width=0.8\linewidth, height=0.1\paperheight]{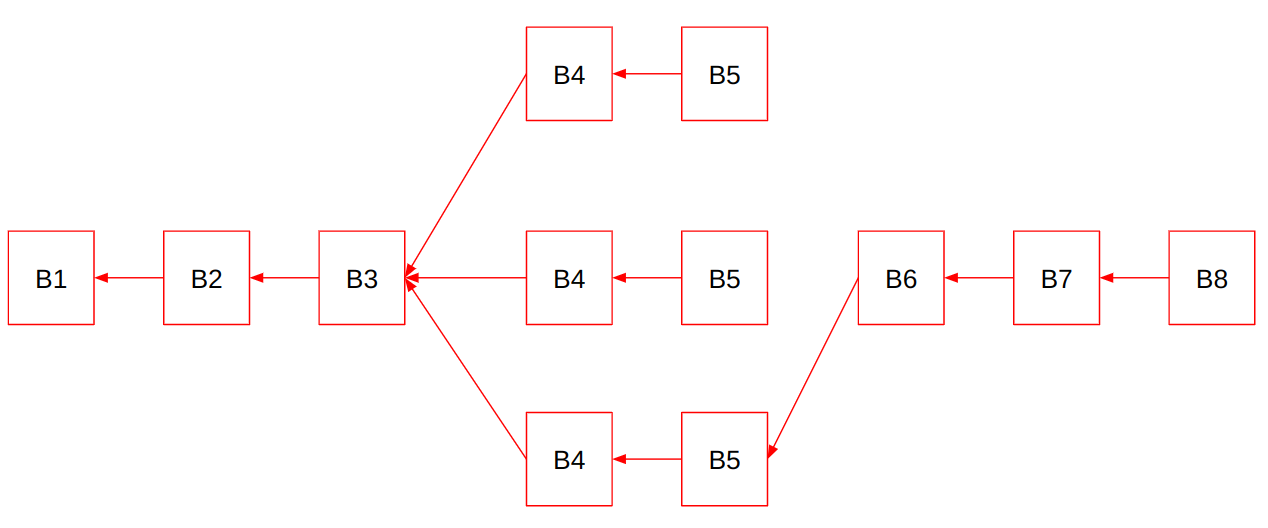}
    \caption{Fallback mode of \sys.}
    \label{fig:blockchain}
    \vspace{-5mm}    
\end{figure}

We introduce \sys, a novel crash fault-tolerant randomized consensus algorithm. \sys achieves three important goals: 
(1) robustness against adversarial network conditions,   
(2) performance comparable to existing leader-based protocols under synchronous network conditions, and
(3) simplicity in design, enabling rapid implementation and seamless integration into existing cloud software stack.

Given the wide adoption of Raft~\cite{ongaro2014search} in cloud software~\cite{kogias2020hovercraft,hashicorp_raft}, we
design \sys by extending Raft.
Building \sys on Raft as the foundation facilitates its rapid adoption in cloud environments, because existing cloud software stacks already support the programming abstractions and infrastructure required by Raft.

As illustrated in~\Cref{fig:blocks}, similar to Raft~\cite{ongaro2014search}, each replica in \sys maintains a sequential log of blocks.
Each block points to the previous block using a pointer.
The purpose of \sys protocol is to ensure that 
each replica commits the same value for each block, despite replica and network failures.

Similar to Raft, \sys employs a leader-based 
design during the synchronous network conditions.
The leader replica in \sys collects a batch of client 
commands, and replicates them 
in at least a majority of the replicas.
Replication under synchronous network conditions
takes one network round trip, the theoretical optimal.
As long as the network conditions remain synchronous
and the leader replica doesn't crash, 
\sys commits commands in one network round trip.

Under adversarial network conditions, Raft 
employs a view change sub-routine, where each replica
votes for a potential new leader replica.
As shown in~\Cref{fig:hangover}, view change based leader
election harms the performance.
In \sys, we eliminate view change-based leader election, marking a fundamental difference from Raft.
Instead of view-change based leader election, \sys 
employs a novel randomized fallback mode
that keeps committing new client commands.
Replacing view change sub routine with the 
randomized fallback mode of operation is the key
difference between Raft and \sys.

\Cref{fig:blockchain} illustrates the randomized fallback mode of \sys from the perspective of a single replica. In~\Cref{fig:blockchain}, 
the network remains synchronous until block B3
is replicated. 
Due to an adversarial network condition, 
the leader replica fails to transmit B4, 
and therefor all the replicas "automatically" enter the randomized fallback mode
of operation.
During the randomized fallback mode, each replica, independent from
other replicas, act as a leader, and 
proposes new blocks.
As shown in~\Cref{fig:blockchain}, in a 3 replica configuration,
the randomized fallback mode results in 3 parallel chains of blocks.
The number of blocks proposed during the randomized fallback mode
can be configured dynamically, but for the purpose of illustration, in this example we use two fallback blocks.
Once each replica has proposed 2 fallback blocks, they consistently chose a single 
chain (the last fallback chain in~\Cref{fig:blockchain}) out of all available fallback chains.
If adversarial conditions persist, another round of fallback mode will be triggered. 
Otherwise, the synchronous path leader of the next view will replicate blocks in a single network round trip, extending the chosen fallback chain.

Replacing view change in Raft with 
a randomized fallback mode helps \sys achieve the three goals mentioned above.
First, because new commands are committed in the randomized fallback mode, \sys remains robust under adversarial network conditions.
Second, during the synchronous periods of execution,
\sys uses a leader-based design, and incurs the same overhead as Raft.
Finally, because \sys builds on Raft,
it is easy to understand and requires fewer additions to the existing Raft implementations~\cite{paxos-raft-code,hashicorp_raft}, enabling rapid adoption in the cloud.
\sys is the first protocol to simultaneously achieve all 3 goals. 

\subsection{\sys Algorithmic Challenges}

\sys has a simple design but faces non-trivial challenges in ensuring agreement and termination.
First, depending on the network conditions, \sys must automatically switch between the synchronous single-chain mode and the fallback multi-chain mode, without manual intervention.
Second, \sys must ensure that each replica commits the same chain of blocks, even when multiple chains exist during adversarial periods -- each replica should consistently choose the same fallback chain at the end of the fallback mode.
\Cref{subsec:racsdescription} addresses these challenges
and presents the \sys algorithm.

\begin{figure*}
    \centering    \includegraphics[height=0.65\paperheight, width=0.77\linewidth]{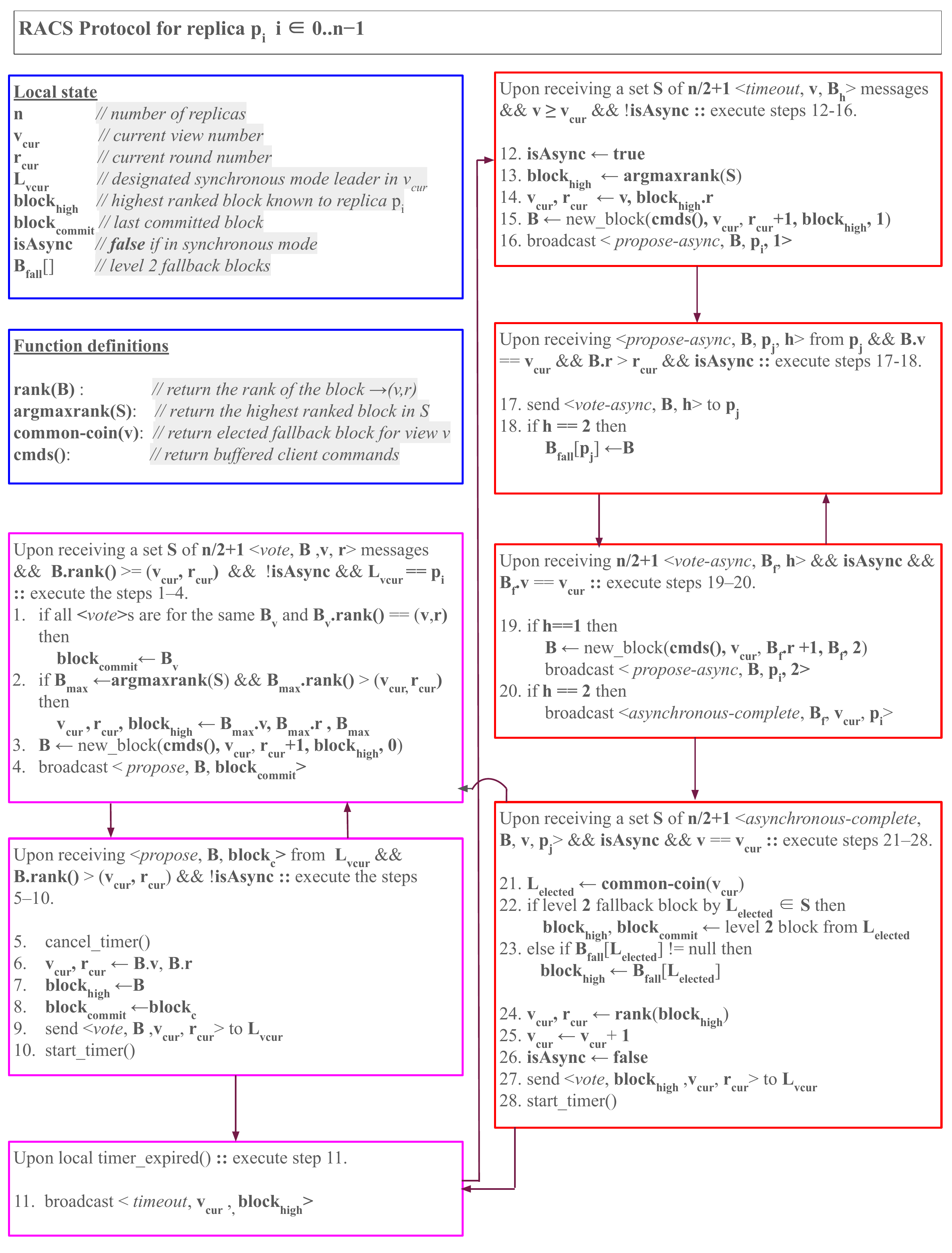}
    \caption{\sys Algorithm.}
    \label{fig:racs_algo}
    \vspace{-5mm}    
\end{figure*}

\subsection{\sys Algorithm}
\label{subsec:racsdescription}

Similar to Raft, \sys progresses as a sequence of \textbf{views} $v$
and \textbf{rounds} $r$ where each view has one or more rounds.
A view consists of (1) a synchronous mode with a predefined leader \( L_{\text{vcur}} \) and (2) a fallback mode of operation.
In~\Cref{fig:blockchain} blocks $B1$--$B5$
are from the same view, where blocks $B1$-$B3$ are synchronous mode blocks and blocks $B4$-$B5$ are fallback mode blocks.
Blocks $B6$--$B8$
are from the next view.
A round represents the successive log positions in the replicated log.
Round numbers are monotonically increasing.
The pair $(v,r)$ is called a \textbf{rank}.

\textbf{Block Format}: There are two types of
\sys blocks: (1) \textbf{synchronous blocks}
and (2) \textbf{fallback blocks}.
Both types of blocks consist of five fields:
(1) batch of client commands,
(2) view number,
(3) round number,
(4) parent block,
and (5) level.
The rank of a block is ($v, r$)
and blocks are compared lexicographically by their rank:
first by the view number, then by the round number.
The blocks are connected in a chain 
using the parent links.
We denote that block $A$ \textbf{extends}
block $B$
if there exists a set of blocks
$b\textsubscript{1}, b\textsubscript{2},
b\textsubscript{3}, .. b\textsubscript{k}$ such that there
exists a parent link
from $b\textsubscript{i}$ to $b\textsubscript{i-1}$
$\forall$ $i$ in $range(2, k)$
and $b\textsubscript{1}$ = $B$
and $b\textsubscript{k}$ = $A$.
The level field of the block
refers to the fallback level.
For the synchronous blocks, the level is always set to $0$.


\cref{fig:racs_algo} depicts the pseudo-code of \sys. 
In the following discussion, when we say replica $p_i$
delivers a block $B$ from replica $p_j$,
we imply that replica $p_i$ delivers $B$ and the
causal history of $B$.


\textbf{Synchronous mode}: The synchronous mode of \sys
is leader-based and is
similar to Raft, as depicted in
lines 1--11 of~\Cref{fig:racs_algo}..
The synchronous mode leader $L_{vcur}$
for each view $v$ is
predetermined and known to all replicas on bootstrap,
for example $L_{vcur} = v_{cur} \% n $.
As we discuss below, $L_{vcur}$ begins its tenure
only after receiving a quorum of $<$\textit{vote}$>$ messages,
which allows $L_{vcur}$ to learn the most up-to-date block—a condition indirectly achieved in Raft~\cite{ongaro2014search} by electing the most up-to-date replica during view change.

The synchronous mode begins either
at the very start of the protocol
or after a fallback mode has ended.
Upon collecting a $n/2+1$ number of $<$\textit{vote}$>$
messages, 
the leader replica 
forms a new block $B$,
that extends the $block_{high}$ (see line 3 of~\Cref{fig:racs_algo}).
The leader then
broadcasts a $<$\textit{propose}$>$
message for the block $B$
and the
reference of the last committed block $block_{commit}$
(see line 4 of~\Cref{fig:racs_algo}).

Each replica $p_i$ delivers $<$\textit{propose}, B, $block_c$$>$ message,
if the rank of $B$ is greater than the rank of $p_i$
and if $p_i$ is in the synchronous mode
of operation.
If these two conditions are met,
then $p_i$ first updates its
$block_{high}$, $v_{cur}$, $r_{cur}$ 
(see line 6--7 of~\Cref{fig:racs_algo})
and then 
commits the block $block_c$
(see line 8 of~\Cref{fig:racs_algo}).
$p_i$ then sends $p_i$'s $<$\textit{vote}$>$ for $B$ to $L_{vcur}$ (see line 9 of~\Cref{fig:racs_algo}).

Upon receiving $n/2+1$ $<$\textit{vote}$>$ messages for $B$ (see line 1 of~\Cref{fig:racs_algo}),
the leader replica commits $B$ (and the causal history). 
As long as the synchronous mode
leader $L_{vcur}$ is responsive,
the synchronous mode will continue,
committing a new block
in one network round trip.


Each replica has a timeout clock
which is reset whenever the replica
receives a new $<$\textit{propose}$>$ message (see line 5 of~\Cref{fig:racs_algo}).
If the timeout expires (see line 11 of~\Cref{fig:racs_algo}),
however,
they will broadcast a $<$\textit{timeout}$>$ message containing the $block_{high}$.

\textbf{Fallback mode}: 
The number of blocks proposed during 
the fallback mode is a tunable parameter
that we explain in~\Cref{subsec:tuningfallback}.
For simplicity of illustration, we use 2 fallback
blocks per view in the following description.
Upon receiving $n/2+1$ $<$\textit{timeout}$>$ messages,
\sys enters the fallback mode
of operation.
In the fallback mode,
all replicas act as leaders.
Each replica takes the highest $block_{high}$
they are aware of,
forms a level $1$ fallback block $B$
with a monotonically increasing rank
compared to the highest $block_{high}$
it received and sends a $<$\textit{propose-async}$>$
message (see line 12--16 of~\Cref{fig:racs_algo}).

Upon receiving a $<$\textit{propose-async}$>$ message from $p_j$,
each replica $p_i$ sends
back a $<$\textit{vote-async}$>$ message to $p_j$
if the rank of the proposed level $1$ block
is greater than the highest-ranked block witnessed
so far (see line 17 of~\Cref{fig:racs_algo}).

Upon receiving $n/2+1$ $<$\textit{vote-async}$>$ messages
for the level $1$ fallback block $B_f$,
each replica will send a level $2$
fallback block $B$ (see line 19 of~\Cref{fig:racs_algo}). 
The algorithm allows catching up to a
higher ranked block by building upon another
replica's level $1$ block.
This is meant to ensure termination
for replicas that fall behind.
All replicas, upon receiving a $<$\textit{propose-async}$>$ message 
for a level 2 fallback block from $p_j$
send a $<$\textit{vote-async}$>$ message to $p_j$
(see line 17--18 of~\Cref{fig:racs_algo}).

Once $n/2+1$ $<$\textit{vote-async}$>$s have been gathered
for the level $2$ fallback block $B_{f}$, 
each replica $p_i$
broadcasts an $<$\textit{asynchronous-complete}$>$ message (see line 20 of~\Cref{fig:racs_algo}).

Once $n/2 + 1$ replicas have submitted $<$\textit{asynchronous-complete}$>$ messages, the fallback mode is complete. The remaining task is to consistently select one chain from the set of fallback chains.
This step is crucial for maintaining agreement. If two replicas select different fallback chains, the system loses the guarantee that all replicas commit the same sequence of blocks.
To ensure that each replica
consistently chooses the same chain, 
we employ a \textbf{common-coin},
an existing technique~\cite{ben1983another,pan2021rabia} used in distributed algorithms.

\textbf{Common-coin}: 
For each view $v$, common-coin($v$)
returns a positive integer in the range ($0$, $n-1$)
where $n$ is the total number of replicas.
The invocation of common-coin($v$) for a given view $v$, at each replica returns
the same value.

\sys uses a common-coin mechanism, where each replica runs a consistent random number generator with a pre-agreed seed.
When used with the same seed, consistent random number generation protocols produce the same output sequence. 
Common-coin enables \sys to consistently choose a fallback chain after completing fallback mode.

Upon receiving $n/2+1$ $<$\textit{asynchronous-complete}$>$ messages,
each replica individually flips a common-coin
(see line 21 of~\Cref{fig:racs_algo})
to get the elected leader L\textsubscript{elected}.
Each replica commits a level $2$ fallback block
$B$ from L\textsubscript{elected}
if $B$ arrived amongst the
first $n/2+1$ $<$asynchronous-complete$>$ messages (see line 22 of~\Cref{fig:racs_algo}).
If a replica observes that 
the level $2$ fallback block from L\textsubscript{elected} 
does not appear in the first $n/2+1$ $<$asynchronous-complete$>$ messages, but appears in B\textsubscript{fall}[L\textsubscript{elected}], then, the replica
sets block\textsubscript{high} to B\textsubscript{fall}[L\textsubscript{elected}]
(see line 23 of~\Cref{fig:racs_algo}). 
All replicas then exit the
fallback mode and resume the synchronous
path by uni-casting a $<$\textit{vote}$>$ message
to the synchronous leader L\textsubscript{vcur} of
the next view (see line 27 of~\Cref{fig:racs_algo}).

\sys fallback mode ensures liveness under adversarial network conditions; it only requires receiving at least $n/2+1$ messages in each asynchronous step. 
This quorum condition is trivially satisfied under the standard assumption that fewer than $n/2$ nodes can fail and with reliable point-to-point links (TCP). 
Thus, delayed messages or timeout settings do not impact the asynchronous liveness guarantee of \sys.

\textbf{\sys protocol summary}: Based on the above \sys explanation, its outcomes are summarized as follows. 
\sys ensures robustness against adversarial network conditions by committing new commands during the randomized fallback mode.
\sys commits client commands in one network round trip, and hence provides optimal synchronous normal case latency.
Finally, because \sys follows a simple design and extends Raft, it enables seamless integration with the existing cloud software stack. 
Although the \sys fallback mode modifies Raft~\cite{ongaro2014search} in non-trivial ways, the fallback mode employs the same replica and network structures, without requiring any additional cloud software stack changes. 
Thus, \sys remains fully compatible with existing cloud software stack and requires only incremental integration effort.

\subsection{Correctness and Complexity}\label{subsec:correctness-intu}

The formal proofs have been
deferred~\Cref{sec:formal-proof}. 


    \textbf{Agreement property}: \sys satisfies the agreement property by ensuring that 
    if a block $B$ is committed in 
    any round $r$, then all the blocks
    with round $r' \geq r$ will extend $B$.


    

    
    \textbf{Termination property}: \sys ensures termination under both synchronous and adversarial
    network conditions, due to common-coin based randomization.

    

\textbf{Complexity} The synchronous mode of \sys
has a linear message and bit complexity
for committing a block.
The fallback mode of \sys has
a complexity of $O(n^{2})$.

%% file: optimizations.tex
\section{\sys Optimizations}
\label{sec:optimizations}

\subsection{Tuning fallback blocks}
\label{subsec:tuningfallback}

In~\Cref{sec:design-racs}, we described \sys with two fallback blocks per view. In practice, fixed blocks per view are suboptimal under unpredictable adversarial conditions.

In \sys, we dynamically adjust the number of blocks in fallback mode using Multi-Armed Bandit techniques~\cite{li2020multi}.
We bootstrap \sys with 2 blocks per fallback mode, and each replica monitors $\textit{s}$: the number of times the algorithm switches between synchronous and fallback mode.
A high \textit{s} implies a network where adversarial conditions persist, whereas a low \textit{s} indicates a transient network problem.
\sys employs a Multi-Armed Bandit-based explore/exploit strategy, increasing the number of blocks per fallback mode when $\textit{s}$ increases and decreasing it when $\textit{s}$ decreases.
\sys uses the same technique as Raft's reconfiguration to consistently configure \textit{s} across all replicas. 
Further details are omitted due to space.

\subsection{Application specific optimizations}

Cloud applications that employ consensus protocols can be broadly categorized into two groups based on their performance objectives:  
(1) \textbf{Class 1}: applications that require moderate throughput and low latency~\cite{burrows2006chubby}~\cite{etcd}~\cite{hunt2010zookeeper} and  
(2) \textbf{Class 2}: applications that require high throughput and can tolerate moderately high latency~\cite{kreps2011kafka}~\cite{apachehadoop}.

For class 1 applications, we employ \textit{pipelining}, 
a well-established technique in distributed algorithms~\cite{pasinduQuePaxa2023,lamport2001paxos}.
In pipelined-\sys, the synchronous mode leader $L_{\text{vcur}}$ proposes multiple blocks before receiving $n/2 + 1$ $<$\textit{vote}$>$ messages,
allowing multiple blocks to be committed within a single round-trip network latency.

\subsubsection{Optimization for class 2 applications}

\begin{figure}
    \centering    \includegraphics[height=0.35\paperheight,width=0.75\linewidth]{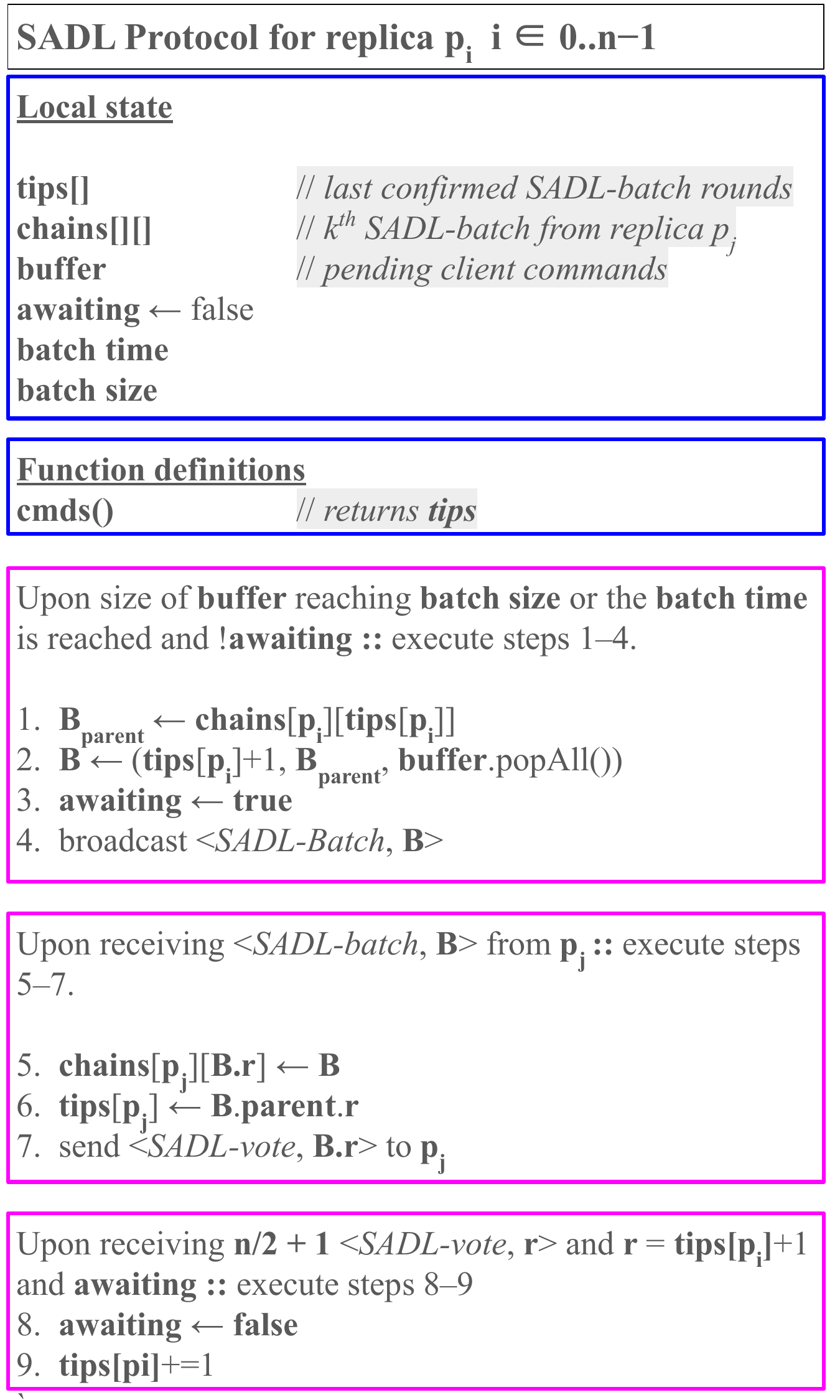}
    \caption{\mempool Algorithm.}
    \label{algo:mempoolalgorithm}
    \vspace{-5mm}
\end{figure}

\sys blocks contain two types of data: (1) metadata about the block such as $v$, and $r$, and (2) a batch of client commands.
Typically, the metadata is only a few bytes in size, while the batch of client commands can span several kilobytes.
Bundling both metadata and client commands within \sys blocks increases the block size, which limits the system's throughput due to the bandwidth bottleneck at the leader replica.
This is a significant limitation of \sys, particularly for class 2 applications that demand high throughput.

To address this limitation, we propose \mempool, a novel command dissemination layer that decouples command dissemination from the critical path of the consensus protocol.
\mempool is a consensus-agnostic command dissemination layer designed to reliably distribute client commands to at least a majority of replicas, ensuring efficient and fault-tolerant command propagation.
\mempool provides fault tolerance and availability for individual client batches, but does not guarantee a total order of commands.
Once \mempool replicates commands across a majority of replicas, \sys layer no longer needs to include the commands in it's blocks. 
Instead, \sys only embeds a small, fixed-size "reference" to the batch already disseminated by \mempool.
Using a fixed-size reference instead of the entire command batch significantly reduces the \sys leader replica's bandwidth consumption, leading to a substantial improvement in throughput, as demonstrated in~\Cref{sec:eval}.

\mempool design significantly differs from existing dissemination approaches.
Unlike existing dissemination mechanisms~\cite{danezis22narwhal,spiegelman2022bullshark}, where the dissemination layer advances in a lock-step fashion, \mempool allows each replica to build its own chain of \mempool blocks independently of the pace of other replicas.
This key difference enables \mempool to achieve higher throughput than existing approaches. We present this subtle yet non-trivial distinction as a novel systems contribution of this paper.

\textbf{\textit{\mempool Protocol}}: 
\Cref{algo:mempoolalgorithm} depicts the \mempool pseudo code.
In \mempool, each replica concurrently proposes \mempool batches. 
For clarity of presentation, we describe the algorithm from the perspective of a single replica $p_i$.
Upon receiving a batch of commands, replica $p_i$ creates a new \mempool batch \textbf{B}, extending its last created \mempool batch $B_{parent}$, and broadcasts a $<$\textit{\mempool-Batch, $\textbf{B}$}$>$ message (see lines 1--4 of~\Cref{algo:mempoolalgorithm}).
Other replicas, upon receiving a new $<$\textit{\mempool-Batch}, \textbf{B}$>$, store \textbf{B} in \textbf{chains}, update \textbf{tips} to the round number of the parent of \textbf{B}, and send a $<$\textit{\mempool-vote}$>$ to $p_i$
(see lines 5--7 of~\Cref{algo:mempoolalgorithm}).
Replica $p_i$, upon receiving $n/2 + 1$ $<$\textit{\mempool-vote}$>$ messages for its last proposed block \textbf{B}, updates the \textbf{tips}
(see line 9 of~\Cref{algo:mempoolalgorithm})
and starts proposing a new \mempool-batch (see lines 1--4 of~\Cref{algo:mempoolalgorithm}).

Replica $p_i$ broadcasts the next \mempool batch only after receiving $n/2 + 1$ votes for the previous \mempool batch it proposed.
Hence, any replica, upon receiving a new \mempool batch $B$ from $p_i$, can be certain that all \mempool batches from $p_i$, prior to $B$ have already been replicated in at least a majority of the replicas.
This property of \mempool enables it to reference a large batch of client commands using a fixed-size integer array of several bytes.
The \textbf{tips} integer array of size \textbf{n} stores the last successfully committed \mempool batch for each replica, thereby allowing the entire set of client batches to be represented using a fixed-size integer array.

\textbf{Using \mempool with \sys}:
With \mempool in place, \sys blocks no longer need to carry the heavy client batches and instead include only the \textbf{tips} integer array as the proposal for consensus.
Upon committing a block in \sys, \sys commits all uncommitted batches in all \mempool chains up to the last block pointed to by the \textbf{tips} array, effectively committing a consistent cut of the \mempool chains.
Since \textbf{tips} is an array of \textbf{n} elements, the size of \sys blocks is drastically reduced, effectively eliminating the leader bottleneck problem in the \sys synchronous mode.
As we show in~\Cref{sec:eval}, \mempool significantly increases throughput, albeit with a modest increase in latency.
Therefore, for class 2 applications that require high throughput but can tolerate moderately high latency, \mempool is well-suited.

%% file: implementation.tex
\section{Implementation}\label{sec:impl}
We implemented pipelined-\sys and \mempool-\sys using Go version 1.19~\cite{go-lang},
in 3661 and 4631 lines of codes, respectively.
We used the standard Go network library and Protobuf encoding~\cite{protobuf}. 
We used a single-threaded event based design for \sys's main event handling logic. 
All events—message receipt, timeouts, etc.—are handled by a single thread, ensuring deterministic event execution. 
For better I/O performance, we use multiple threads for TCP message reception, as done in prior work~\cite{pan2021rabia, pasinduQuePaxa2023, paxos-raft-code}. Both \sys and \mempool-\sys implement batching in both clients and replicas. 
SADL-RACS implementation is publicly available on github\footnote{\href{SADL-RACS github}{https://github.com/ISTA-SPiDerS/sadl-racs}}.

%% file: evaluation.tex
\section{Experimental evaluation}\label{sec:eval}

This evaluation demonstrates the following 4 claims.

\begin{compactitem}
    \item \textbf{C1}: \sys offers robustness against adversarial network conditions.
    \item \textbf{C2}: Under synchronous network conditions, \sys performs comparably to existing leader-based algorithms.
    \item \textbf{C3}: \mempool-\sys delivers high throughput for class 2 applications that demand high throughput.
    \item \textbf{C4}: \mempool improves the scalability of \sys w.r.t (\textbf{C4.1}) increasing replica count and (\textbf{C4.2}): increasing payload size.
\end{compactitem}

\vspace{2mm}

Since adversarial networks 
are much more common in the wide-area network (WAN) than in the local-area
network (LAN),
we focus on the
WAN deployments in our evaluation,
however, for completeness of experiments, we also compared the performance of \sys
in a LAN in \cref{sec:eval-sys-lan-performance}.

We compare \sys's and \mempool-\sys's performance against
three state-of-the-art SMR algorithms:
Raft~\cite{ongaro2014search}
\cite{paxos-raft-code},
Multi-Paxos~\cite{lamport2001paxos}
\cite{paxos-raft-code}, 
and
EPaxos~\cite{moraru2013there}
\cite{epaxos-code-modified}. 
Raft and Multi-Paxos are the most widely used consensus algorithms in cloud environments.
EPaxos is a multi-leader protocol
that enables parallel commits of non-interfering commands.

\textbf{Setup}:
\label{subsec:evaluation-setup}
We test both
a WAN setup
where the replicas and clients are distributed globally
across AWS cloud regions Sydney, Tokyo, Seoul, Osaka, and Singapore and
a LAN setup
where all replicas and clients are located in North California AWS cloud. 

For WAN experiments, 
we chose t2.xlarge EC2 instances~\cite{amazon-ec2} with 4 vCPUs and 16 GB memory.
For LAN experiments, we use instances
of type c4.4\-xlarge (16 virtual CPUs, 30 GB memory)
for replicas and clients.
Since CPU is often the bottleneck in LAN settings—where network latency is low and available bandwidth is high—it is standard to use machines with more CPUs for LAN experiments and we followed this established practice.
To ensure a fair, apples-to-apples comparison, 
we ran Raft, Multi-Paxos, and EPaxos on the same testbed. 
We use Ubuntu Linux 20.04.5 LTS~\cite{ubuntu}.

\textbf{Workload and Benchmarks}:
Following the existing implementations Multi-Paxos
and Raft \cite{paxos-raft-code},
we use a $map[string]string$ key-value store and Redis\cite{go-redis} 
as backend applications.

In our experiments, we use $n$ replicas and $n$ clients.
Clients generate client commands with a Poisson distribution
in the open-loop model~\cite{schroeder2006open}.
All algorithms employ batching in both clients and replicas.
A single client command is a 17 bytes string:
1-byte GET/PUT opcode plus 8-byte keys and values,
consistent with command sizes
used in prior research and production
systems~\cite{pan2021rabia, bronson2013tao}.

As the initial leader replica of Multi-Paxos~\cite{lamport2001paxos},
Raft~\cite{ongaro2014search}, and \sys, we selected the replica located in Tokyo AWS region. 
We tested with other initial leader replica locations (Sydney, Seoul, Osaka, and Singapore), however, observed no noticeable performance difference. 
In Multi-Paxos and Raft, once a view timeout occurs, the next leader is selected uniformly at random -- whoever initiates and successfully completes the view-change sub-routine becomes the leader for the subsequent view. 
In \sys, we use a deterministic round-robin next leader scheme for each new view.

For \sys, \mempool-\sys, Multi-Paxos, and Raft we measure the
client observed end-to-end execution latency, which accounts for the 
latency overhead for total ordering and executing commands.
EPaxos provides two modes of operations: (1) partial ordering of commands without execution (denoted ``EPaxos-commit'' in the graphs) and 
(2) partial ordering of commands
with execution (denoted ``EPaxos-exec'' in the graphs).
Trivially, ``EPaxos-commit''
outperforms \sys and \mempool-\sys
because ''EPaxos-commit'' only provides a partial
order of commands, which enables higher parallelism.
Hence, ''EPaxos-commit'' provides an apples-to-oranges comparison, however, we 
present the results in this evaluation, for completeness.
We also found and reported bugs in the existing implementation of
EPaxos code
that prevent execution under adversarial network conditions, crashes,
and when deployed with more 
than 5 replicas.
Hence we use EPaxos only under normal-case
performance evaluation.

We measure throughput in commands per second (cmd/sec),
where a command is one 17-byte command.
We measure the latency in milliseconds.

\input{evaluation-asynchrony}
\input{evaluation-normal-wan}

\input{evaluation-scalabliity}
\input{evaluation-normal-case-lan}

%% file: evaluation-asynchrony.tex
\begin{figure}[t]
    \centering
    \includegraphics[scale=0.45]{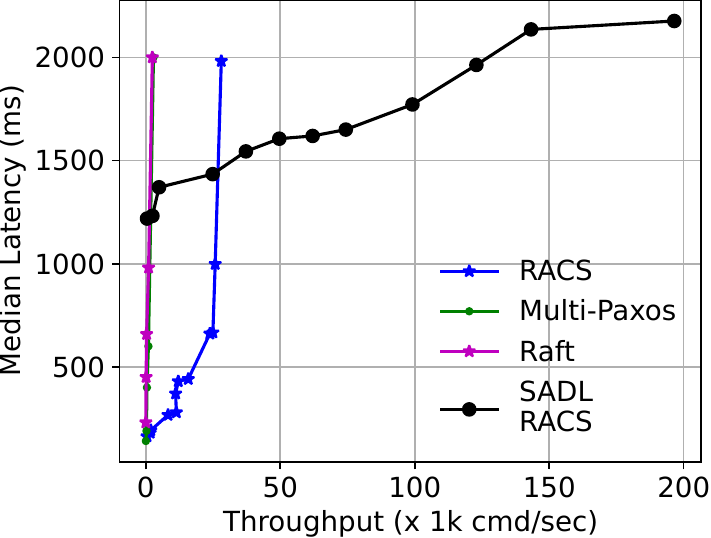}
    \caption{Adversarial Performance in the WAN with 5 replicas.}
    \label{fig:asynchrony}
    \vspace{-6mm}
\end{figure}

\subsection{Asynchronous Performance}\label{subsec:eval-asynchronous}

A methodology to comprehensively evaluate consensus protocol performance
under asynchronous network conditions remains an open research problem~\cite{tennageconsenstress}.
The only published method for asynchronous evaluation appears in~\cite{pasinduQuePaxa2023}, which we follow in this paper.

This experiment evaluates \sys and \mempool-\sys
under simulated network attack scenarios.
The attacker increases the egress packet latency of a randomly
selected minority of replicas by 500ms in dynamic time epochs.
The experiment runs in a WAN setting with 5 replicas and 5 clients.
We measured throughput, latency, downtime, and time-to-recover, but due to space limits and to align with recent literature~\cite{mahi-mahi, pasinduQuePaxa2023}, 
we report only throughput and latency.
We depict the results in \Cref{fig:asynchrony}.

\textbf{\sys vs \{Multi-Paxos and Raft}\}:We observe that
\sys provides 28k cmd/sec saturation throughput,
in contrast, Multi-Paxos and Raft 
have saturation throughput at 2.8k cmd/sec.
Under adversarial network conditions 
Multi-Paxos and Raft undergo repeated view changes,
and fail at successfully committing commands.
In contrast, due to the fallback mode of execution,
\sys provides liveness under adversarial network conditions.
Hence we prove the claim \textbf{C1}.

\textbf{\sys vs \mempool-\sys}:
We observe that \mempool-\sys delivers 196k saturation throughput,
thus providing 
168k cmd/sec more throughput than \sys.
Because \mempool-\sys eliminates the leader bottleneck and the \mempool protocol is robust against adversarial network conditions (as it does not rely on timeouts), \mempool-\sys achieves higher throughput than \sys, even under adversarial network conditions.

%% file: evaluation-normal-wan.tex
\begin{figure*}[t]
    \centering
    \begin{subfigure}[b]{0.24\linewidth}
        \centering
        \includegraphics[width=0.9\textwidth]{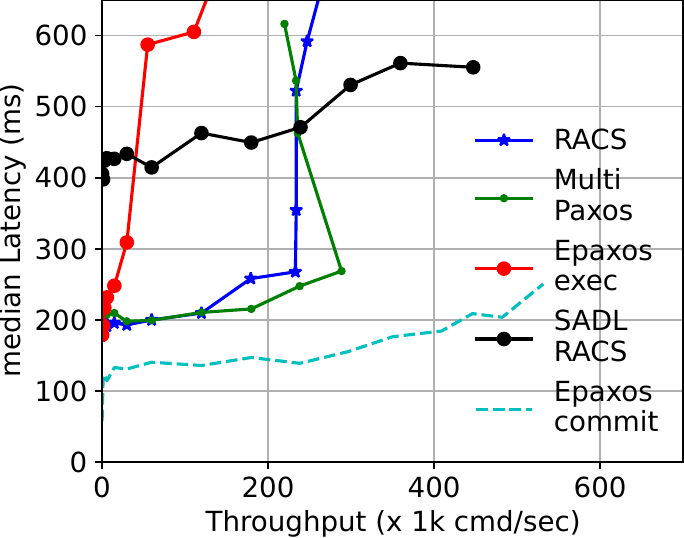}
        \caption{Median Latency 3 replicas}
        \label{fig:normal_case_performance_wan_median_3_replicas}
        
    \end{subfigure}
    \begin{subfigure}[b]{0.24\linewidth}
        \centering
        \includegraphics[width=0.9\textwidth]{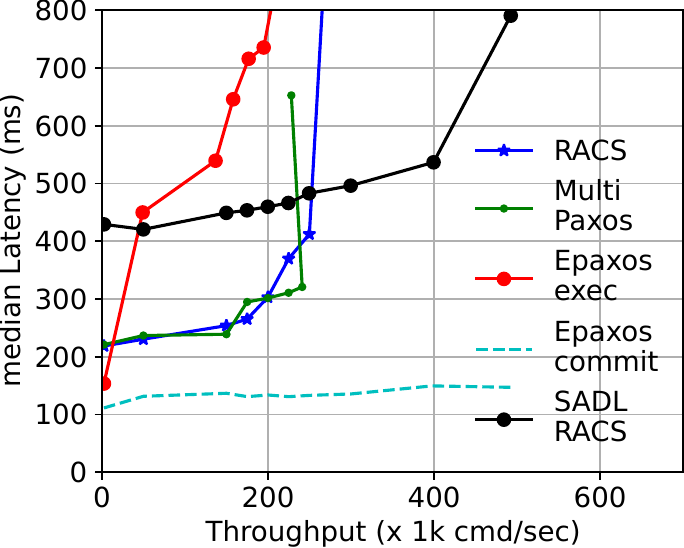}
        \caption{Median Latency 5 replicas}
        \label{fig:normal_case_performance_wan_median_5_replicas}
        
    \end{subfigure}
    \begin{subfigure}[b]{0.24\linewidth}
        \centering
        \includegraphics[width=0.9\textwidth]{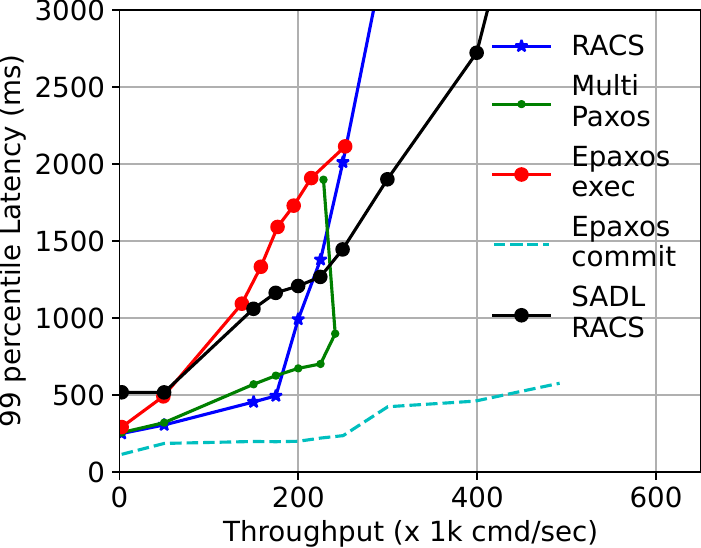}
        \caption{99\% Latency 5 replicas}
        \label{fig:normal_case_performance_wan_tail_5_replicas}
        
    \end{subfigure}    
    \begin{subfigure}[b]{0.24\linewidth}
        \centering
        \includegraphics[width=0.9\textwidth]{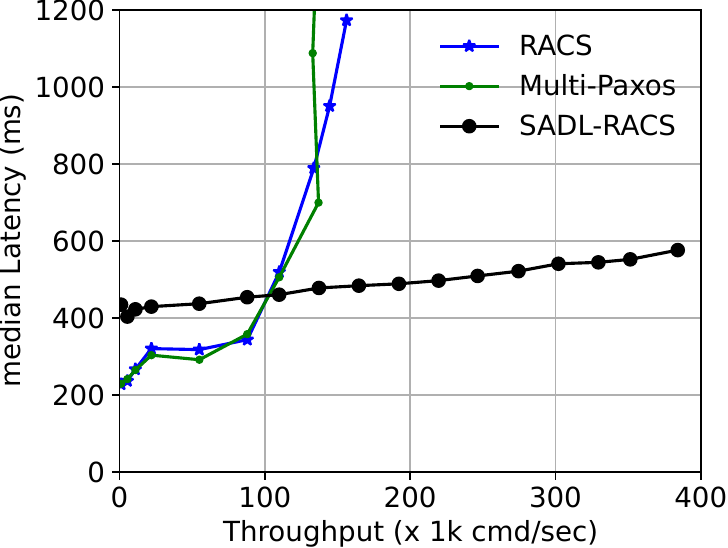}
        \caption{Median Latency 11 replicas}
        \label{fig:normal_case_performance_wan_median_11_replicas}
        
    \end{subfigure}
    
    \caption{Throughput versus latency for WAN normal-case execution for 3--11 number of replicas.
    }
    \label{fig:normal_case_wan_performance_replicas}
    \vspace{-4mm}    
\end{figure*}

\subsection{\sys WAN Normal Case Performance}
\label{subsec:eval-wan-racs-performance}

In this experiment, we evaluate the
normal-case synchronous performance of \sys
deployed in 5 geographically distant AWS regions.
Because we observed the same throughput and latency for Raft and Multi-Paxos (due to their identical structure during normal-case replication), we only depict the Multi-Paxos performance for clarity in the figures.
\cref{fig:normal_case_performance_wan_median_5_replicas} and \cref{fig:normal_case_performance_wan_tail_5_replicas}
depict the experimental results using 5 replicas
and 5 clients.

\textbf{\sys vs Multi-Paxos}: We observe 
in \cref{fig:normal_case_performance_wan_median_5_replicas} 
that \sys delivers a saturation throughput of 200k cmd/sec throughput
under 300ms median latency, which is comparable to
the performance of Multi-Paxos (200k cmd/sec under 300ms latency).
In the synchronous execution both \sys and Multi-Paxos
have 1 round trip latency per batch of commands,
hence share the same performance characteristics.
Hence the experimental claim \textbf{C2} holds.

\textbf{\sys vs Epaxos-commit:} We observe in 
~\cref{fig:normal_case_performance_wan_median_5_replicas}
that EPaxos-commit (without command execution)
delivers a throughput of 500k+ cmd/sec under 170ms median latency.
The EPaxos-commit experiment employs a conflict rate of 
2\%\cite{tollman2021epaxos}
hence 98\% of the time, commands are 
committed in one round trip, without 
serializing through a leader replica.
In contrast, \sys builds a total order of commands,
serialized using a leader-replica,
hence naturally the \sys performance is bottlenecked by the leader 
replica's capacity.

\textbf{\sys vs Epaxos-exec:}
As shown in \cref{fig:normal_case_performance_wan_median_5_replicas},
the median latency of EPaxos-exec (with command execution)
is 300ms higher on average than \sys
in the 50k--200k cmd/sec throughput range.
This higher latency stems
from EPaxos's dependency management
cost~\cite{tollman2021epaxos,matte2021scalable}.
Hence, we conclude that when measured for execution latency, 
\sys outperforms EPaxos.

EPaxos is suited for applications that require only a partial order, 
such as key-value stores where per-key ordering is sufficient.
In contrast, \sys is designed for applications that need a total order of all commands.

\textbf{\sys vs \mempool-\sys:}
As shown in \cref{fig:normal_case_performance_wan_median_5_replicas},
\mempool-\sys achieves a throughput of 500k cmd/sec, exceeding \sys by 300k cmd/sec, 
however, with an additional latency cost of 500ms. 
These results demonstrate the effectiveness of \mempool in eliminating the leader bottleneck in \sys for class 2 applications that require high throughput but can tolerate moderately high latency, supporting the claim \textbf{C3}.

%% file: evaluation-scalabliity.tex
\begin{figure*}[t]
    \centering
    \begin{subfigure}[b]{0.24\linewidth}
        \centering
        \includegraphics[width=0.9\textwidth]{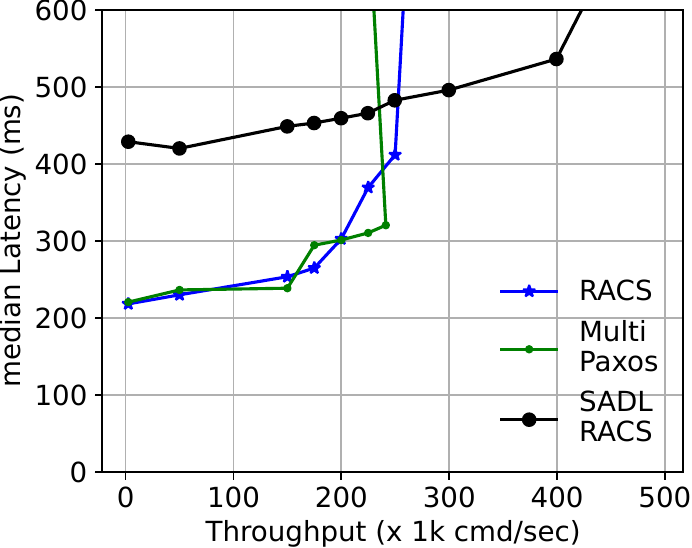}
        \caption{Median Latency (17B)}
        \label{fig:normal_case_performance_wan_median_17_payload}
        
    \end{subfigure}
    \begin{subfigure}[b]{0.24\linewidth}
        \centering
        \includegraphics[width=0.9\textwidth]{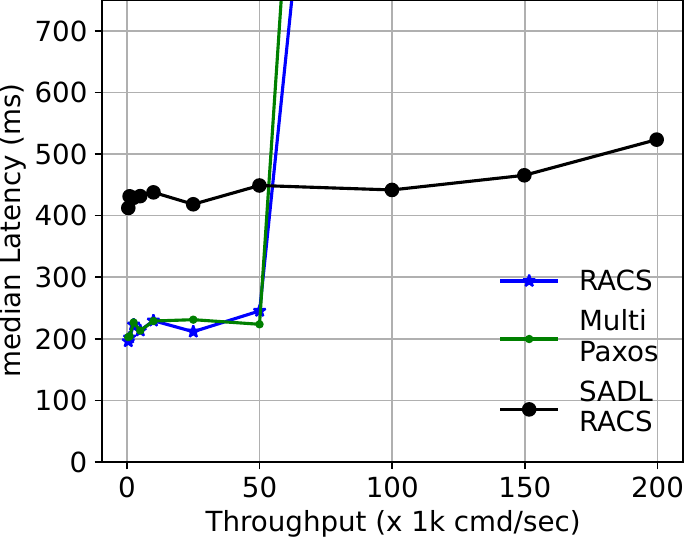}
        \caption{Median Latency (73B)}
        \label{fig:normal_case_performance_wan_median_73_payload}
        
    \end{subfigure}
    \begin{subfigure}[b]{0.24\linewidth}
        \centering
        \includegraphics[width=0.9\textwidth]{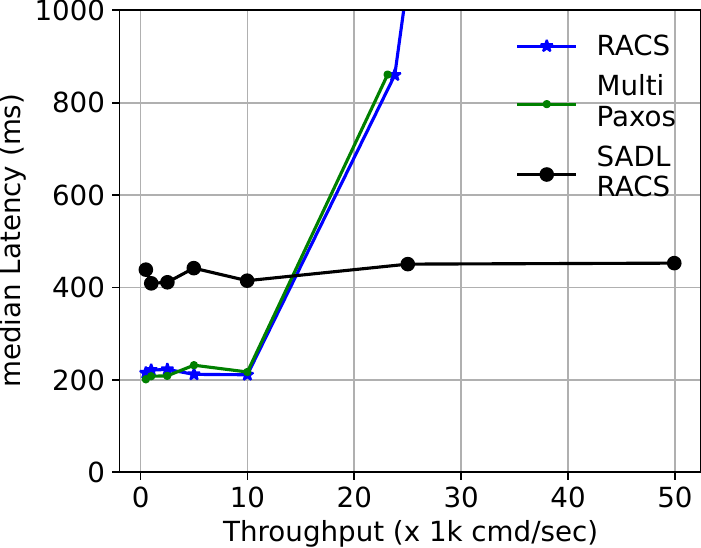}
        \caption{Median Latency (265B)}
        \label{fig:normal_case_performance_wan_median_265_payload}
        
    \end{subfigure}
    \begin{subfigure}[b]{0.24\linewidth}
        \centering
        \includegraphics[width=0.9\textwidth]{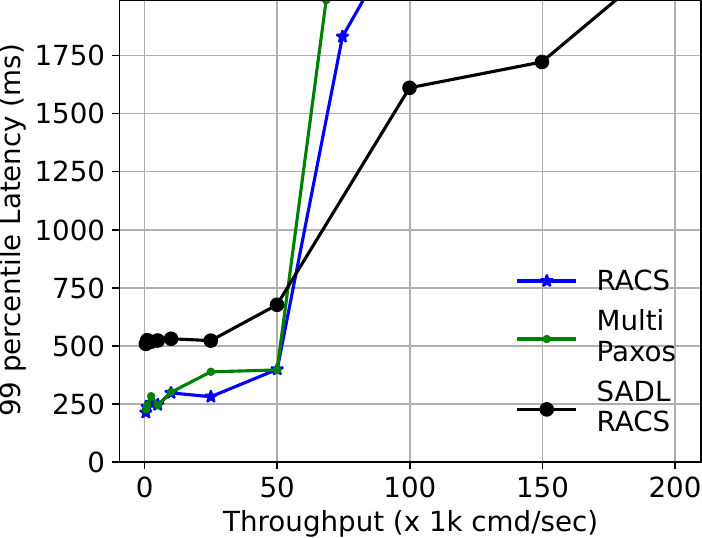}
        \caption{99\% Latency (73B)}
        \label{fig:normal_case_performance_wan_tail_73_payload}
    \end{subfigure}
    \caption{Throughput versus latency for WAN normal-case execution using 17B, 73B and 265B command sizes.}
    \label{fig:normal_case_wan_message_size_performance}
    \vspace{-3mm}    
\end{figure*}

\subsection{Scalability of \mempool}
\label{subsec:eval-scalability}


In this experiment, we aim to quantify the scalability of \mempool-\sys.
We consider two factors of scalability:
(1) scalability w.r.t increasing replication factor 
and (2) scalability w.r.t increasing payload size.

Although we conducted experiments with both Multi-Paxos and Raft, 
we observed the same throughput and latency behavior due to their identical protocol structure during normal-case replication 
(both replicate in a single network round trip, serialized through a leader). 
Therefore, for clarity in figures and explanation, we present only the Multi-Paxos results.

\textbf{Scalability w.r.t increasing replication factor}: In this experiment,
we evaluate the scalability of \mempool-\sys 
by running it with an ensemble of three (minimum replication allowed),
five (common replication factor) and eleven replicas (improved robustness to concurrent
replica failures), located in geographically 
separated AWS regions.
\cref{fig:normal_case_wan_performance_replicas} compare the scalability of \mempool-\sys against pipelined \sys, pipelined Multi-Paxos, and pipelined EPaxos, for different replication factors.

\textbf{\sys vs \{Multi-Paxos and Raft\}}: We observe that the saturation 
throughput of Multi-Paxos and \sys decreases from 230k 
to 130k cmd/sec (under 600ms median latency) when the replication factor is increased from 
3--11.
With increasing replica count, 
the leader replica in \sys and Multi-Paxos has to
send and receive more messages, due to increased
quorum sizes, hence the 
performance is bottlenecked by the leader's bandwidth capacity.

\textbf{\sys vs \mempool-\sys}: We observe that \mempool-\sys
provides a throughput of 380k cmd/sec (under 600 ms median latency),
when the replication factor is 11.
\mempool-\sys outperforms pipelined-\sys and pipelined Multi-Paxos by 
192\% in the 11 replica scenario.
This confirms that \mempool enables \sys to increase throughput by avoiding the leader bottleneck through the decoupling of command dissemination from consensus.
Hence we prove the claim \textbf{C4.1}.

\textbf{Scalability w.r.t increasing payload size}:
In this experiment, we evaluate the impact of payload 
size for the \mempool-\sys performance.
We experiment with 3 key sizes:
8B, 64B, and 256B, 
used in recent SMR work~\cite{alimadadi2023waverunner}.
Combined with 1B opcode and
8B value, these key sizes result in
17B, 73B, and 265B command sizes.
We deploy \mempool-\sys, pipelined Paxos and pipelined-\sys
in a WAN setting with 5 replicas and 5 clients.
\cref{fig:normal_case_wan_message_size_performance} depicts 
the results.

We observe that for each command size, the 
saturation throughput of \mempool-\sys is 
at least 2 times the throughput of \sys and Multi-Paxos.
With increasing command size, the leader replica's
bandwidth of \sys (and Multi-Paxos)
becomes the bottleneck.
In contrast,
thanks to the decoupling of command dissemination
from consensus, \mempool-\sys evenly distributes the 
bandwidth overhead among all the replicas,
and sustains higher throughput.
This proves the claim \textbf{C4.2}.

%% file: evaluation-normal-case-lan.tex
\subsection{\sys LAN Normal Case Performance}\label{sec:eval-sys-lan-performance}

\begin{figure}
    \centering
    \includegraphics[width=0.3\textwidth]{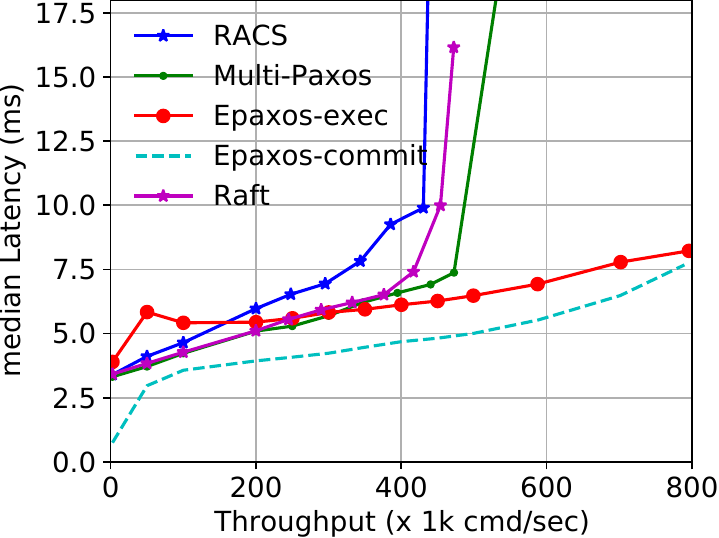}    
    \caption{Throughput versus latency for LAN normal-case.}
    \label{fig:normal_case_lan_performance}
    \vspace{-5mm}
\end{figure}

In this experiment, we quantify the performance of the \sys in a LAN setting.
We conduct experiments with 5 replicas and 5 clients deployed in the same AWS region. \Cref{fig:normal_case_lan_performance} illustrates the LAN performance of \sys.

\textbf{\sys vs \{Multi-Paxos and Raft\}}:
We observe that \sys achieves a saturation throughput of
420k cmd/sec, under a median
latency upper bound of 10ms,
that is comparable to the saturation throughput of Multi-Paxos (450k) and
Raft (440k).
Under normal case executions,
\sys, Multi-Paxos, and Raft have 1 round-trip
latency, serialized through a leader,
hence provide comparable performance.


%% file: conclusion.tex
\section{Conclusion, Limitations and Future Work}\label{sec:conc}

We presented \sys, the first  randomized consensus protocol that simultaneously achieves three key goals; 
(1) robustness under adversarial network conditions, 
(2) performance comparable to existing leader-based protocols, under synchronous network conditions, and
(3) simple design inspired by Raft, enabling seamless integration.

\mempool-\sys has limitations that require future work. 
\mempool-\sys requires at least four message hops per commit, making it less suitable for low-latency, high-throughput applications. 
We propose a DAG-based redesign that merges dissemination and ordering as a plausible method to achieve both low latency and high throughput. 

While \mempool-\sys optimizes for performance, robustness, and simplicity, it does not address energy efficiency—its chain-based design commits blocks even at low request rates, leading to high energy use. 
Although energy-efficient cloud systems are well studied, consensus protocol energy use remains unexplored, opening future directions in defining energy metrics, analyzing existing protocols, and designing energy-efficient consensus protocols. 

Finally, \sys assumes that a majority of replicas are connected, however, loses liveness under extreme network partitions with only one or two connected replicas. 
Extending \sys to support such network partitions is left for future work.

\subsubsection*{Acknowledgements.} This work was supported by (1) Mysten Labs and the Sui Foundation, (2) the Austrian Science Fund (FWF) through the SFB SpyCode project F8512-N and by the WWTF through the project 10.47379/ICT22045, (3) EPFL research grants, (4) AXA Research Fund, and (5) Etherium grant.

%% file: racs_proof.tex
\section{\sys Formal Proofs}\label{sec:formal-proof}

\subsection{Definition} 
\textbf{elected-fallback block}: 
We refer to an fallback block $B_f$ generated in view
$v$ with level $2$ as an elected-fallback block,
if the common-coin-flip($v$) returns the index of
the proposer $p_l$ who generated $B_f$ in the
view $v$ and if the $<$asynchronous-complete$>$
for $B_f$ exists in the first $n-f$
$<$asynchronous-complete$>$ messages received.

\subsection{Proof of agreement}

\begin{theorem}\label{th:1}
    Let $B$ and $\tilde{B}$ be two blocks with rank $(v,r)$. Each of $B$ and $\tilde{B}$ can be of type: (1) synchronous block which collects at least $n-f$ votes or (2) elected-fallback block or (3) level $1$ fallback block which is a parent of an elected-fallback block. Then $\tilde{B}$ and $B$ are the same.
\end{theorem}
\begin{proof}
    This holds directly from the block formation -- if both $B$ and $\tilde{B}$ has the same rank, then due to quorum intersection, there exists at least one node who voted for both blocks in the same rank, which is a contradiction to our assumption of non malicious nodes.
\end{proof}

\begin{theorem}\label{th:2}
    Let $B$ and $\tilde{B}$ be two adjacent blocks, then $\tilde{B}.r = B.r+1$ and $\tilde{B}.v \geq B.v$.
\end{theorem}
\begin{proof}
    According to the algorithm, there are three instances where a new block is created.

\begin{itemize}
    \item Case 1: when $isAsync$ = false and $L_{vcur}$ creates a new
    synchronous block by extending the $block_{high}$ with rank ($v$,$r$). In this case, $L_{vcur}$ creates a new block with round $r+1$. Hence the adjacent blocks have monotonically increasing round numbers.
    \item Case 2: when $isAsync$ = true and upon collecting $n-f$ $<$timeout$>$ messages in view $v$. In this case, the replica selects the $block_{high}$ with the highest rank $(v,r)$, and extends it by proposing a level $1$ fallback block with round $r+1$. Hence the adjacent blocks have monotonically increasing round numbers.
    \item Case 3: when $isAsync$ = true and upon collecting $n-f$ $<$vote-async$>$ messages for a level $1$ fallback block. In this case, the replica extends the level $1$ block by proposing a level $2$ block with round $r+1$. Hence the adjacent blocks have monotonically increasing round numbers.
\end{itemize}

The view numbers are non decreasing according to the algorithm. Hence Theorem~\ref{th:2} holds.
\end{proof}

\begin{theorem}\label{th:3}
    If a synchronous block $B_c$ with rank $(v,r)$ is committed, then all future blocks
    in view $v$ will extend $B_c$.
\end{theorem}
\begin{proof}
    We prove this by contradiction. 

    Assume there is a committed block $B_c$ with $B_c$.$r$ = $r_c$ (hence all the blocks in the path from the genesis block to $B_c$ are committed). Let block $B_s$ with $B_s$.$r$ = $r_s$ be the round $r_s$ block such that $B_s$ conflicts with $B_c$ ($B_s$ does not extend $B_c$). Without loss of generality, assume that $r_c < r_s$.

    Let block $B_f$ with $B_f$.$r$ = $r_f$ be the first valid block formed in a round $r_f$ such that $r_s \geq r_f > r_c$ and $B_f$ is the first block from the path from genesis block to $B_s$ that conflicts with $B_c$; for instance $B_f$ could be $B_s$. $L_{vcur}$ forms
    $B_f$ by extending its $block_{high}$.
    Due to the minimality of $B_f$ ($B_f$ is the first block that conflicts with $B_c$), $block_{high}$ contain either $B_c$ or a block that extends $B_c$. Since $block_{high}$ extends $B_c$, $B_f$ extends $B_c$, thus we reach a contradiction. Hence no such $B_f$ exists. Hence all the blocks created after $B_c$ in the view $v$ extend $B_c$.
\end{proof}

\begin{theorem}\label{th:4}
    If a synchronous block $B$ with rank $(v,r)$ is committed, an elected-fallback block $\tilde{B}$ of the same view $v$ will extend that block.
\end{theorem}
\begin{proof}
    We prove this by contradiction. Assume that a synchronous block $B$ is committed in view $v$ and an elected-fallback block $\tilde{B}$ does not extend $B$. Then, the parent level $1$ block of $\tilde{B}$, $\tilde{B_p}$, also does not extend $B$.

    To form the level $1$ $\tilde{B_p}$, the replica collects $n-f$ $<$timeout$>$ messages, each of them containing the $block_{high}$. If $B$ is committed, by theorem~\ref{th:3}, at least $n-f$ replicas should have set (and possibly sent) $B$ or a block extending $B$ as the $block_{high}$. Hence by intersection of the quorums $\tilde{B_p}$ extends $B$, thus we reach a contradiction.
\end{proof}

\begin{theorem}\label{th:5}
    At most one level 2 fallback block from one proposer can be committed in a given view change.
\end{theorem}
\begin{proof}
    Assume by way of contradiction that 2 level $2$ fallback blocks from two different proposers are committed in the same view. A level $2$ fallback block $B$ is committed in the fallback phase if the common-coin-flip($v$) returns the proposer of $B$ as the elected proposer. Since the common-coin-flip($v$) outputs the same elected proposer across different replicas, this is a contradiction. Thus all level $2$ fallback blocks committed during the same view are from the same proposer.
    
    Assume now that the same proposer proposed two different level $2$ fallback blocks. Since no replica can equivocate, this is absurd.
    
    Thus at most one level $2$ fallback block from one proposer can be committed in a given view change.
\end{proof}

\begin{theorem}\label{th:6}
    Let $B$ be a level $2$ elected-fallback block that is committed, then all blocks proposed in the subsequent rounds extend $B$.
\end{theorem}
\begin{proof}
    We prove this by contradiction. Assume that level two elected-fallback block $B$ is committed with rank $(v,r)$ and block $\tilde{B}$ with rank ($\tilde{v}$, $\tilde{r}$) such that ($\tilde{v}$, $\tilde{r}$) $>$ $(v,r)$ is the first block in the chain starting from $B$ that does not extend $B$. $\tilde{B}$ can be formed in two occurrences: (1) $\tilde{B}$ is a synchronous block in the view $v+1$ or (2) $\tilde{B}$ is a level $1$ fallback block with a view strictly greater than $v$. (we do not consider the case where $\tilde{B}$ is a level 2 elected-fallback block, because this directly follows from \ref{th:1})

    If $B$  is committed, then from the algorithm construction it is clear that a majority of the replicas will set $B$ as $block_{high}$. This is because, to send a $<$asynchronous-complete$>$ message with $B$, a replica should collect at least $n-f$ $<$vote-async$>$ messages. Hence, its guaranteed that if $\tilde{B}$ is formed in view $v$+1 as a synchronous block, then it will observe $B$ as the $block_{high}$, thus we reach a contradiction. 
    
    In the second case, if $\tilde{B}$ is formed in a subsequent view, then it is guaranteed that the level $1$ block will extend $B$ by gathering from the $<$timeout$>$ messages $B$ as $block_{high}$ or a block extending $B$ as the $block_{high}$, hence we reach a contradiction.
\end{proof}

\begin{theorem}\label{th:7}
    There exists a single history of committed blocks.
\end{theorem}
\begin{proof}
    Assume by way of contradiction there are two different histories $H_1$ and $H_2$ of committed blocks. Then there is at least one block from $H_1$ that does not extend at least one block from $H_2$. This is a contradiction with theorems~\ref{th:3},~\ref{th:4} and~\ref{th:6}. Hence there exists a single chain of committed blocks.
\end{proof}

\begin{theorem}\label{th:8}
    For each committed replicated log position $r$, all replicas contain the same block.
\end{theorem}
\begin{proof}
    By theorem~\ref{th:2}, the committed chain will have incrementally increasing round numbers. Hence for each round number (log position), there is a single committed entry, and by theorem~\ref{th:1}, this entry is unique. This completes the proof.
\end{proof}

\subsection{Proof of termination}

\begin{theorem}\label{th:9}
    If at least $n-f$ replicas enter the fallback phase of view $v$ by setting $isAsync$ to true, then eventually they all exit the fallback phase and set $isAsync$ to false.
\end{theorem}
\begin{proof}
    If $n-f$ replicas enter the fallback path, then eventually all replicas (except for failed replicas) will enter the fallback path as there are less than $n-f$ replicas left on the synchronous path due to quorum intersection, so no progress can be made on the synchronous path and all replicas will timeout. As a result, at least $n-f$ correct replicas will broadcast their $<$timeout$>$ message and all replicas will enter the fallback path.
    
    Upon entering the fallback path, each replica creates a fallback block with level $1$ and broadcasts it. Since we use perfect point-to-point links, eventually all the level $1$ blocks sent by the $n-f$ correct replicas will be received by each replica in the fallback path. At least $n-f$ correct replicas will send them $<$vote-async$>$ messages if the rank of the level $1$ block is greater than the rank of the replica. To ensure liveness for the replicas that have a lower rank, the algorithm allows catching up, so that nodes will adopt whichever level $1$ block which received $n-f$ $<$vote-async$>$ arrives first. Upon receiving the first level $1$ block with $n-f$ $<$vote-async$>$ messages, each replica will send a level $2$ fallback block, which will be eventually received by all the replicas in the fallback path. Since the level $2$ block proposed by any block passes the rank test for receiving a $<$vote-async$>$, eventually at least $n-f$ level $2$ blocks get $n-f$ $<$vote-async$>$. Hence, eventually at least $n-f$ replicas send the $<$asynchronous-complete$>$ message, and exit the fallback path.
\end{proof}

\begin{theorem}\label{th:10}
    With probability $p>\frac{1}{2}$, at least one replica commits an elected-fallback block after exiting the fallback path.
\end{theorem}
\begin{proof}
    Let leader $L_{elected}$ be the output of the common-coin-flip($v$). A replica commits a block during the fallback mode if the $<$asynchronous-complete$>$ message from $L_{elected}$ is among the first $n-f$ $<$asynchronous-complete$>$ messages received during the fallback mode, which happens with probability at least greater than $\frac{1}{2}$. Hence with probability no less than $\frac{1}{2}$, each replica commits a chain in a given fallback phase.
\end{proof}

\begin{theorem}\label{th:11}
    A majority of replicas keep committing new blocks with high probability.
\end{theorem}
\begin{proof}
    We first prove this theorem for the basic case where all replicas start the protocol with $v = 0$. If at least $n-f$ replicas eventually enter the fallback path, by theorem~\ref{th:9}, they eventually all exit the fallback path, and a new block is committed by at least one replica with probability no less than $\frac{1}{2}$. According to the asynchronous-complete step, all nodes who enter the fallback path enter view $v = 1$ after exiting the fallback path. If at least $n-f$ replicas never set $isAsync$ to true, this implies that the sequence of blocks produced in view $1$ is infinite. By Theorem~\ref{th:2}, the blocks have consecutive round numbers, and thus a majority replicas keep committing new blocks.

    Now assume the theorem~\ref{th:11} is true for view $v = 0,..., k-1$. Consider the case where at least $n-f$ replicas enter the view $v = k$. By the same argument for the $v = 0$ base case, $n-f$ replicas either all enter the fallback path commits a new block with $\frac{1}{2}$ probability, or keeps committing new blocks in view $k$. Therefore, by induction, a majority replicas keep committing new blocks with high probability.
\end{proof}

\begin{theorem}\label{th:12}
    Each client command is eventually committed.
\end{theorem}
\begin{proof}
    If each replica repeatedly keeps proposing the client commands until they become committed, then eventually each client command
    gets committed according to theorem~\ref{th:11}.
\end{proof}

%% file: main.bbl
\begin{thebibliography}{10}
\providecommand{\url}[1]{#1}
\csname url@samestyle\endcsname
\providecommand{\newblock}{\relax}
\providecommand{\bibinfo}[2]{#2}
\providecommand{\BIBentrySTDinterwordspacing}{\spaceskip=0pt\relax}
\providecommand{\BIBentryALTinterwordstretchfactor}{4}
\providecommand{\BIBentryALTinterwordspacing}{\spaceskip=\fontdimen2\font plus
\BIBentryALTinterwordstretchfactor\fontdimen3\font minus
  \fontdimen4\font\relax}
\providecommand{\BIBforeignlanguage}[2]{{%
\expandafter\ifx\csname l@#1\endcsname\relax
\typeout{** WARNING: IEEEtran.bst: No hyphenation pattern has been}%
\typeout{** loaded for the language `#1'. Using the pattern for}%
\typeout{** the default language instead.}%
\else
\language=\csname l@#1\endcsname
\fi
#2}}
\providecommand{\BIBdecl}{\relax}
\BIBdecl

\bibitem{cachin2011introduction}
C.~Cachin, R.~Guerraoui, and L.~Rodrigues, \emph{Introduction to Reliable and
  Secure Distributed Programming}.\hskip 1em plus 0.5em minus 0.4em\relax
  Springer Science \& Business Media, 2011.

\bibitem{burrows2006chubby}
M.~Burrows, ``The chubby lock service for loosely-coupled distributed
  systems,'' in \emph{7th Symposium on Operating Systems Design and
  Implementation}, 2006, pp. 335--350.

\bibitem{hunt2010zookeeper}
P.~Hunt, M.~Konar, F.~P. Junqueira, and B.~Reed, ``$\{$ZooKeeper$\}$: Wait-free
  coordination for internet-scale systems,'' in \emph{2010 USENIX Annual
  Technical Conference (USENIX ATC 10)}, 2010.

\bibitem{maccormick2004boxwood}
J.~MacCormick, N.~Murphy, M.~Najork, C.~A. Thekkath, and L.~Zhou, ``Boxwood:
  Abstractions as the foundation for storage infrastructure.'' in
  \emph{Symposium on Operating Systems Design and Implementation OSDI}, vol.~4,
  2004, pp. 8--8.

\bibitem{corbett2013spanner}
J.~C. Corbett, J.~Dean, M.~Epstein, A.~Fikes, C.~Frost, J.~J. Furman,
  S.~Ghemawat, A.~Gubarev, C.~Heiser, P.~Hochschild \emph{et~al.}, ``Spanner:
  Google’s globally distributed database,'' \emph{ACM Transactions on
  Computer Systems (TOCS)}, vol.~31, no.~3, pp. 1--22, 2013.

\bibitem{baker2011megastore}
J.~Baker, C.~Bond, J.~C. Corbett, J.~Furman, A.~Khorlin, J.~Larson, J.-M. Leon,
  Y.~Li, A.~Lloyd, and V.~Yushprakh, ``Megastore: Providing scalable, highly
  available storage for interactive services,'' in \emph{Conference on
  Innovative Data system Research}, 2011, pp. 223--234.

\bibitem{xie2014salt}
C.~Xie, C.~Su, M.~Kapritsos, Y.~Wang, N.~Yaghmazadeh, L.~Alvisi, and
  P.~Mahajan, ``Salt: Combining $\{$ACID$\}$ and $\{$BASE$\}$ in a distributed
  database,'' in \emph{11th USENIX Symposium on Operating Systems Design and
  Implementation (OSDI 14)}, 2014, pp. 495--509.

\bibitem{quintero2011implementing}
D.~Quintero, M.~Barzaghi, R.~Brewster, W.~H. Kim, S.~Normann, P.~Queiroz,
  R.~Simon, A.~Vlad \emph{et~al.}, \emph{Implementing the IBM General Parallel
  File System (GPFS) in a Cross Platform Environment}.\hskip 1em plus 0.5em
  minus 0.4em\relax IBM Redbooks, 2011.

\bibitem{mashtizadeh2013replication}
A.~J. Mashtizadeh, A.~Bittau, Y.~F. Huang, and D.~Mazieres, ``Replication,
  history, and grafting in the ori file system,'' in \emph{Proceedings of the
  Twenty-Fourth ACM Symposium on Operating Systems Principles}, 2013, pp.
  151--166.

\bibitem{grimshaw2013gffs}
A.~Grimshaw, M.~Morgan, and A.~Kalyanaraman, ``Gffs—the xsede global
  federated file system,'' \emph{Parallel Processing Letters}, vol.~23, no.~02,
  p. 1340005, 2013.

\bibitem{bronson2013tao}
N.~Bronson, Z.~Amsden, G.~Cabrera, P.~Chakka, P.~Dimov, H.~Ding, J.~Ferris,
  A.~Giardullo, S.~Kulkarni, H.~Li, M.~Marchukov, D.~Petrov, L.~Puzar, Y.~J.
  Song, and V.~Venkataramani, ``{TAO}: {Facebook}’s distributed data store
  for the social graph,'' in \emph{USENIX Annual Technical Conference USENIX
  ATC 13}, June 2013, pp. 49--60.

\bibitem{lloyd2011don}
W.~Lloyd, M.~J. Freedman, M.~Kaminsky, and D.~G. Andersen, ``Don't settle for
  eventual: Scalable causal consistency for wide-area storage with cops,'' in
  \emph{Proceedings of the Twenty-Third ACM Symposium on Operating Systems
  Principles}, 2011, pp. 401--416.

\bibitem{lamport2001paxos}
L.~Lamport, ``{Paxos} made simple,'' \emph{ACM SIGACT News (Distributed
  Computing Column) 32, 4}, vol.~32, pp. 51--58, December 2001.

\bibitem{ongaro2014search}
D.~Ongaro and J.~Ousterhout, ``In search of an understandable consensus
  algorithm,'' in \emph{2014 USENIX Annual Technical Conference ATC14)}, June
  2014, pp. 305--319.

\bibitem{nastic2020sloc}
S.~Nastic, A.~Morichetta, T.~Pusztai, S.~Dustdar, X.~Ding, D.~Vij, and
  Y.~Xiong, ``Sloc: Service level objectives for next generation cloud
  computing,'' \emph{IEEE Internet Computing}, vol.~24, no.~3, pp. 39--50,
  2020.

\bibitem{ding2019characterizing}
J.~Ding, R.~Cao, I.~Saravanan, N.~Morris, and C.~Stewart, ``Characterizing
  service level objectives for cloud services: Realities and myths,'' in
  \emph{2019 IEEE International Conference on Autonomic Computing
  (ICAC)}.\hskip 1em plus 0.5em minus 0.4em\relax IEEE, 2019, pp. 200--206.

\bibitem{cloudflare-nov20}
T.~Lianza and C.~Snook, ``Cloudflare outage,''
  \url{https://blog.cloudflare.com/a-byzantine-failure-in-the-real-world/},
  November 2020.

\bibitem{moura2016review}
J.~Moura and D.~Hutchison, ``Review and analysis of networking challenges in
  cloud computing,'' \emph{Journal of Network and Computer Applications},
  vol.~60, pp. 113--129, 2016.

\bibitem{de2022noise}
D.~De~Sensi, T.~De~Matteis, K.~Taranov, S.~Di~Girolamo, T.~Rahn, and
  T.~Hoefler, ``Noise in the clouds: Influence of network performance
  variability on application scalability,'' \emph{Proceedings of the ACM on
  Measurement and Analysis of Computing Systems}, vol.~6, no.~3, pp. 1--27,
  2022.

\bibitem{moraru2013there}
I.~Moraru, D.~G. Andersen, and M.~Kaminsky, ``There is more consensus in
  egalitarian parliaments,'' in \emph{Proceedings of the Twenty-Fourth ACM
  Symposium on Operating Systems Principles}, November 2013, pp. 358--372.

\bibitem{barcelona2008mencius}
Y.~Mao, F.~Junqueira, and K.~Marzullo, ``{Mencius}: Building efficient
  replicated state machines for {WAN}s,'' in \emph{8th USENIX Symposium on
  Operating Systems Design and Implementation (OSDI 08)}, December 2008.

\bibitem{charapko2021pigpaxos}
A.~Charapko, A.~Ailijiang, and M.~Demirbas, ``{PigPaxos}: Devouring the
  communication bottlenecks in distributed consensus,'' in \emph{Proceedings of
  the 2021 International Conference on Management of Data}, June 2021, pp.
  235--247.

\bibitem{lamport2005generalized}
\BIBentryALTinterwordspacing
L.~Lamport, ``Generalized consensus and {Paxos},'' \emph{Microsoft Research
  Technical Report MSR-TR-2005-33}, p.~60, 2005. [Online]. Available:
  \url{https://www.microsoft.com/en-us/research/publication/generalized-consensus-and-paxos/}
\BIBentrySTDinterwordspacing

\bibitem{tennage2022baxos}
P.~Tennage, C.~Basescu, E.~K. Kogias, E.~Syta, P.~Jovanovic, and B.~Ford,
  ``Baxos: Backing off for robust and efficient consensus,'' \emph{arXiv
  preprint arXiv:2204.10934}, December 2024.

\bibitem{pasinduQuePaxa2023}
P.~Tennage, C.~Basescu, L.~Kokoris-Kogias, E.~Syta, P.~Jovanovic, V.~Estrada,
  and B.~Ford, ``{QuePaxa}: Escaping the tyranny of timeouts in consensus,''
  \emph{Proceedings of the 29th Symposium on Operating Systems Principles
  (SOSP)}, Oct. 2023.

\bibitem{ben1983another}
M.~Ben-Or, ``Another advantage of free choice (extended abstract) completely
  asynchronous agreement protocols,'' in \emph{Proceedings of the Second Annual
  ACM symposium on Principles of Distributed Computing}, August 1983, pp.
  27--30.

\bibitem{BoBandle}
\BIBentryALTinterwordspacing
B.~Wang, S.~Liu, H.~Dong, X.~Wang, W.~Xu, J.~Zhang, P.~Zhong, and Y.~Zhang,
  ``Bandle: Asynchronous state machine replication made efficient,'' in
  \emph{Proceedings of the Nineteenth European Conference on Computer Systems},
  ser. EuroSys '24.\hskip 1em plus 0.5em minus 0.4em\relax New York, NY, USA:
  Association for Computing Machinery, 2024, p. 265–280. [Online]. Available:
  \url{https://doi.org/10.1145/3627703.3650091}
\BIBentrySTDinterwordspacing

\bibitem{danezis22narwhal}
G.~Danezis, L.~Kokoris-Kogias, A.~Sonnino, and A.~Spiegelman, ``{Narwhal} and
  {Tusk}: A {DAG}-based {Mempool} and efficient {BFT} consensus,'' in
  \emph{Proceedings of the Seventeenth European Conference on Computer Systems
  (EuroSys '22)}, Mar. 2022, pp. 34--50.

\bibitem{go-lang}
J.~Meyerson, ``The {Go} programming language,'' \emph{IEEE Software}, vol.~31,
  no.~5, pp. 104--104, 2014.

\bibitem{rfc793}
``Transmission control protocol,'' Sep. 1981, rFC 793.

\bibitem{aspnes2003randomized}
J.~Aspnes, ``Randomized protocols for asynchronous consensus,''
  \emph{Distributed Computing}, vol.~16, no. 2-3, pp. 165--175, 2003.

\bibitem{rescorla2018transport}
E.~Rescorla and T.~Dierks, ``The transport layer security ({TLS}) protocol
  version 1.3,'' August 2018, rFC 8446.

\bibitem{fischer1985impossibility}
M.~J. Fischer, N.~A. Lynch, and M.~S. Paterson, ``Impossibility of distributed
  consensus with one faulty process,'' \emph{Journal of the ACM (JACM)},
  vol.~32, no.~2, pp. 374--382, 1985.

\bibitem{oki1988viewstamped}
B.~M. Oki and B.~H. Liskov, ``Viewstamped replication: A new primary copy
  method to support highly-available distributed systems,'' in
  \emph{Proceedings of the Seventh Annual ACM Symposium on Principles of
  Distributed Computing}, january 1988, pp. 8--17.

\bibitem{marandi2010ring}
P.~J. Marandi, M.~Primi, N.~Schiper, and F.~Pedone, ``{Ring} {Paxos}: A
  high-throughput atomic broadcast protocol,'' in \emph{IEEE/IFIP International
  Conference on Dependable Systems \& Networks (DSN)}.\hskip 1em plus 0.5em
  minus 0.4em\relax IEEE, June 2010, pp. 527--536.

\bibitem{ng23omni}
H.~Ng, S.~Haridi, and P.~Carbone, ``{Omni-Paxos}: Breaking the barriers of
  partial connectivity,'' in \emph{Eighteenth European Conference on Computer
  Systems (EuroSys)}, May 2023, pp. 314--330.

\bibitem{ailijiang2017wpaxos}
A.~Ailijiang, A.~Charapko, M.~Demirbas, and T.~Kosar, ``{WPaxos}: Wide area
  network flexible consensus,'' \emph{IEEE Transactions on Parallel and
  Distributed Systems}, vol.~31, no.~1, pp. 211--223, 2019.

\bibitem{pan2021rabia}
H.~Pan, J.~Tuglu, N.~Zhou, T.~Wang, Y.~Shen, X.~Zheng, J.~Tassarotti, L.~Tseng,
  and R.~Palmieri, ``{Rabia}: Simplifying state-machine replication through
  randomization,'' in \emph{Proceedings of the ACM SIGOPS 28th Symposium on
  Operating Systems Principles}, October 2021, pp. 472--487.

\bibitem{mahi-mahi}
P.~Jovanovic, L.~Kokoris-kogias, B.~Kumara, A.~Sonnino, P.~Tennage, and
  I.~Zablotchi, ``Mahi-mahi: Low-latency asynchronous bft dag-based
  consensus,'' in \emph{IEEE ICDCS 2025, 45th IEEE International Conference on
  Distributed Computing Systems}, 2025.

\bibitem{keidar2021all}
I.~Keidar, E.~Kokoris-Kogias, O.~Naor, and A.~Spiegelman, ``All you need is
  dag,'' in \emph{Proceedings of the 2021 ACM Symposium on Principles of
  Distributed Computing}, 2021, pp. 165--175.

\bibitem{kogias2020hovercraft}
M.~Kogias and E.~Bugnion, ``{HovercRaft}: Achieving scalability and
  fault-tolerance for microsecond-scale datacenter services,'' in
  \emph{Proceedings of the Fifteenth European Conference on Computer Systems},
  April 2020, pp. 1--17.

\bibitem{hashicorp_raft}
HashiCorp, ``Raft: Golang implementation of the raft consensus protocol,''
  \url{https://github.com/hashicorp/raft}, 2013, accessed: 2025-01-12.

\bibitem{paxos-raft-code}
P.~Tennage, ``\href{https://github.com/dedis/paxos-and-raft}{Paxos and Raft},''
  Sep. 2023, gitHub repository \url{https://github.com/dedis/paxos-and-raft}.

\bibitem{li2020multi}
F.~Li, D.~Yu, H.~Yang, J.~Yu, H.~Karl, and X.~Cheng,
  ``{Multi-armed-bandit-based} spectrum scheduling algorithms in wireless
  networks: A survey,'' \emph{IEEE Wireless Communications}, vol.~27, no.~1,
  pp. 24--30, 2020.

\bibitem{etcd}
{The etcd Authors}, ``etcd: Distributed reliable key-value store,''
  \url{https://etcd.io/}, 2013, accessed: 2024-01-14.

\bibitem{kreps2011kafka}
J.~Kreps, N.~Narkhede, J.~Rao \emph{et~al.}, ``Kafka: A distributed messaging
  system for log processing,'' in \emph{Proceedings of the NetDB}, vol.~11, no.
  2011.\hskip 1em plus 0.5em minus 0.4em\relax Athens, Greece, 2011, pp. 1--7.

\bibitem{apachehadoop}
{Apache Software Foundation}, ``Apache hadoop: Open-source framework for
  distributed storage and processing of large data sets,''
  \url{https://hadoop.apache.org/}, 2006, accessed: 2024-01-14.

\bibitem{spiegelman2022bullshark}
A.~Spiegelman, N.~Giridharan, A.~Sonnino, and L.~Kokoris-Kogias, ``Bullshark:
  Dag bft protocols made practical,'' in \emph{Proceedings of the 2022 ACM
  SIGSAC Conference on Computer and Communications Security}, 2022, pp.
  2705--2718.

\bibitem{protobuf}
Google, ``Protocol buffers,''
  \url{https://developers.google.com/protocol-buffers/}, 2020.

\bibitem{epaxos-code-modified}
I.~Moraru, D.~G. Andersen, M.~Kaminsky, and P.~Tennage, ``{EPaxos} go-lang,''
  \url{https://github.com/dedis/quepaxa-ePaxos-open-loop}, Sep. 2023.

\bibitem{amazon-ec2}
Amazon, ``{AWS} instance types,''
  \url{https://aws.amazon.com/ec2/instance-types/}, 2023.

\bibitem{ubuntu}
Ubuntu, ``{Ubuntu} {Linux},'' \url{https://releases.ubuntu.com/focal/}, 2023.

\bibitem{go-redis}
V.~Mihailenco, Denissenko, and Dimitrij, ``{Go lang} {Redis},''
  \url{https://github.com/redis/go-redis}, 2023.

\bibitem{schroeder2006open}
B.~Schroeder, A.~Wierman, and M.~Harchol-Balter, ``Open versus closed: A
  cautionary tale,'' in \emph{Proceedings of the 3rd USENIX Symposium on
  Networked Systems Design and Implementation (NSDI 06)}.\hskip 1em plus 0.5em
  minus 0.4em\relax USENIX, May 2006.

\bibitem{tennageconsenstress}
P.~Tennage, S.~Mishra, A.~Sonnino, L.~K. Kogias, P.~Jovanovic, and B.~Ford,
  ``Consenstress: A framework to torture test consensus protocols,'' in
  \emph{EuroSys 2025 (poster)}, 2025.

\bibitem{tollman2021epaxos}
S.~Tollman, S.~J. Park, and J.~K. Ousterhout, ``{EPaxos} revisited,'' in
  \emph{USENIX Symposium on Networked Systems Design and Implementation (NSDI
  21)}, April 2021, pp. 613--632.

\bibitem{matte2021scalable}
V.~S. Matte, A.~Charapko, and A.~Aghayev, ``Scalable but wasteful: Current
  state of replication in the cloud,'' in \emph{Proceedings of the 13th ACM
  Workshop on Hot Topics in Storage and File Systems}, July 2021, pp. 42--49.

\bibitem{alimadadi2023waverunner}
M.~Alimadadi, H.~Mai, S.~Cho, M.~Ferdman, P.~Milder, and S.~Mu, ``Waverunner:
  An elegant approach to hardware acceleration of state machine replication,''
  in \emph{20th USENIX Symposium on Networked Systems Design and Implementation
  (NSDI 23)}, 2023, pp. 357--374.

\end{thebibliography}
